\documentclass[11 pt]{article}
\usepackage{amsmath}
\usepackage{amsfonts}
\usepackage{amssymb}
\usepackage{amsthm}
\usepackage{fullpage}
\usepackage{caption}
\usepackage{bbm}
\usepackage{hyperref, color}
\hypersetup{colorlinks=true,citecolor=blue, linkcolor=blue, urlcolor=blue}
\usepackage[linesnumbered,boxed,commentsnumbered,ruled,vlined]{algorithm2e}
\usepackage{bm}
\usepackage{bbm}
\usepackage[numbers]{natbib}
\usepackage{xcolor}
\usepackage{enumerate} 
\usepackage{tabularx}
\usepackage{array}
\newcolumntype{L}[1]{>{\raggedright\arraybackslash}p{#1}}
\newcolumntype{C}[1]{>{\centering\arraybackslash}m{#1}}
\newcolumntype{R}[1]{>{\raggedleft\arraybackslash}p{#1}}

\usepackage{makecell}
\usepackage[T1]{fontenc}
\usepackage{footnote}
\makesavenoteenv{tabular}

\newcommand{\E}[1]{\mathbb{E}\left[{#1}\right]}

\renewcommand{\epsilon}{\varepsilon}

\newcommand{\Expand}{\mathsf{Expand}}
\newcommand{\MExp}{\MC{M}_{\mathsf{Exp}}}

\newcommand{\vbl}[1]{\mathsf{vbl}\left(#1\right)}

\newcommand{\mgn}[1]{\tilde{{#1}}}
\newcommand{\MC}[1]{\mathfrak{{#1}}}
\newcommand{\Resample}{\textsf{Local-Resample}}

\newcommand{\GenResample}{\textsf{GenResample}}
\newcommand{\LRes}{\mathsf{Res}}
\newcommand{\GRes}{\mathsf{GR}}
\newcommand{\MLRes}{\MC{M}_{\mathsf{Res}}}
\newcommand{\HLRes}{H}
\newcommand{\HHC}{H_{\mathsf{HC}}}

\newcommand{\concept}[1]{\emph{{#1}}}

\newcommand{\todo}[1]{\typeout{TODO: \the\inputlineno: #1}\textbf{{\color{red}[[[ #1 ]]]}}}

\newtheorem{theorem}{Theorem}[section]

\newtheorem{claim}[theorem]{Claim}
\newtheorem*{claim*}{Claim}
\newtheorem{condition}{Condition}

\newtheorem{lemma}[theorem]{Lemma}

\theoremstyle{definition}
\newtheorem{definition}{Definition}[section]
\newtheorem{remark}[theorem]{Remark}
\newtheorem*{remark*}{Remark}

\makeatletter
\def\blfootnote{\xdef\@thefnmark{}\@footnotetext}
\makeatother

\title{\bf Dynamic Sampling from Graphical Models\blfootnote{This research is supported by the National Key R\&D Program of China 2018YFB1003202 and the National Science Foundation of China under Grant Nos. 61722207 and 61672275.}
}
\author{
Weiming Feng\thanks{
Nanjing University. Emails: {fengwm@smail.nju.edu.cn}, {yinyt@nju.edu.cn}.}
\and
Nisheeth K. Vishnoi\thanks{\'{E}cole Polytechnique F\'{e}d\'{e}rale de Lausanne (EPFL). Email: nisheeth.vishnoi@gmail.com.}
\and
Yitong Yin\footnotemark[1]~\footnote{Part of the work was done when Yitong Yin was visiting the Bernoulli Center at EPFL.}
}
\date{}

\begin{document}
\maketitle
\begin{abstract}
In this paper, we study the problem of sampling from a graphical model  when the model itself is changing dynamically with time.
This problem derives its interest from a variety of  inference, learning, and sampling settings in machine learning, computer vision, statistical physics, and theoretical computer science. 
While the problem of  sampling from a static graphical model has received considerable attention, theoretical works for its dynamic variants 
have been largely lacking.
The main contribution of this paper is an algorithm that can   sample dynamically from a broad class of graphical models over discrete random variables.
Our algorithm is parallel and Las Vegas: it knows when to stop and it outputs samples from the exact distribution.
We also provide sufficient conditions under which this algorithm runs in time proportional to the size of the update,
on general graphical models as well as well-studied specific spin systems.
In particular we obtain, for the Ising model (ferromagnetic or anti-ferromagnetic) and for the hardcore model the first dynamic sampling algorithms that can handle both edge and vertex updates (addition, deletion, change of functions), 
both efficient within regimes that are close to the respective uniqueness regimes, beyond which, even for the static and approximate sampling, no local algorithms were known or the problem itself is intractable.
Our dynamic sampling algorithm relies on a  local resampling algorithm and  a  new ``equilibrium'' property that is shown to be satisfied by our algorithm at each step, and enables us to prove its correctness.
This equilibrium property is robust enough to guarantee the correctness of our algorithm, 
helps us improve bounds on fast convergence on specific models, and should be of independent interest.
\end{abstract}

\setcounter{page}{0} \thispagestyle{empty} \vfill
\pagebreak

\tableofcontents
\setcounter{page}{0} \thispagestyle{empty} \vfill
\pagebreak

\section{Introduction}

Graphical models arise in a variety of disciplines ranging from statistical physics, machine learning, statistics, to theoretical computer science.
A graphical model is composed of a variable set $V$ and a constraint set $E$. 
We consider graphical models on variables with finite support.
In this setting, each variable $v \in V$ is associated to a distribution $\phi_v$ over the set $[q]=\{1,2,\ldots,q\}$.
Each constraint $e \in E$ is a subset of variables and comes with a function $\phi_e:[q]^{e} \to [0,1]$ defined on the variables in $e$. 
Together, these induce a  probability distribution $\mu$ over all possible assignments  $\sigma \in [q]^V$ as follows:
$$ \mu(\sigma) \propto \prod_{v \in V} \phi_v(\sigma_v) \prod_{e \in E} \phi_e(\sigma_e),$$
where $\sigma_v$ (respectively $\sigma_e$) corresponds to the restriction of $\sigma$ on $v$ (respectively $e$).
This distribution is often refered to as the Gibbs distribution.
Such graphical models  can capture probability distributions over exponentially sized domains in a succinct manner.
The computational problems that arise from the application of graphical models in practice include sampling from the probability distribution they encode, computing marginals (inference), and learning a graphical model; see the books by
 \cite{koller2009probabilistic, mezard2009information, wainwright2008graphical}.
These problems often turn out to be computationally hard in the worst case and there is a wide range of methods  geared towards solving these problems approximately: 
 Markov chain Monte Carlo (MCMC) methods~\cite{levin2017markov}, correlation decay~\cite{weitz2006counting}, belief propagation~\cite{yedidia2005constructing}, and continuous optimization~\cite{straszak2017real}.

We focus on the problem of sampling from a graphical model and, in particular, when the model itself is changing {\em dynamically} with time. 
For instance, at each time, one or more of the functions $\phi_v$ or $\phi_e$ could change.
Formally, the computational question that we study is:

\begin{center}
{\em Can we obtain a sample from an updated graphical model with a small incremental cost?}
\end{center}
\noindent
This problem captures various settings in computer vision, statistical physics, and machine learning. 
In computer vision, discrete-valued graphical models are used to represent images and the problem of denoising an image boils down to sampling from such a graphical model. 
Thus, algorithms that can sample dynamically with a small incremental cost are useful in denoising videos, which can be thought of as a sequence of closely related images; see \cite{bishop2006PRML}. 
As another instance of this question, consider the setting where one uses an optimization algorithm such as stochastic gradient descent (or expectation maximization), to learn a graphical model~\cite{jordan1998learning}. 
Here, the gradient step updates the parameters of the model locally and samples from the updated distribution are used to compute the new gradient.
In theoretical computer science, the result of Jerrum, Sinclair, and Vigoda~\cite{jerrum2004polynomial} to compute the permanent of a non-negative matrix can also be viewed in this framework: their algorithm starts from a bipartite graph where it is easy to sample a perfect matching and, in each step the bipartite graph is updated.
While the problem of  sampling from static graphical models has received considerable attention, theoretical works for its dynamic variant that work with general graphical models have been largely lacking.
Indeed, since a local update may potentially change significantly the probability space encoded by the graphical model, 
in general it is unclear whether there should even exist such an algorithmic machinery that can transform with small cost a sample from the original probability space to a new sample from the updated probability space.

In this paper we show there exists such an algorithmic machinery for dynamic sampling from graphical models.
The main contribution of this paper is an algorithm that allows us to sample from a broad class of graphical models dynamically.
We allow updates (addition, deletion, and changes to the functions) to the variables and constraints.
Given a sample from the current graphical model, our sampling algorithm outputs a sample from the updated graphical model.
In addition, the algorithm is parallel and Las Vegas; it knows when to stop and it outputs samples from the {\em exact} distribution.
We also provide sufficient conditions under which the algorithm runs in time proportional to the size of the update.
This gives the first dynamic sampling algorithm that can handle both variable and constraint updates (addition, deletion, change of functions) for various graphical models.
In particular, for the Ising model with inverse temperature $\beta$ (ferromagnetic or anti-ferromagnetic) and bounded maximum degree $\Delta$ under the condition $\mathrm{e}^{-2 |\beta|} \ge 1- \frac{1}{2.222\Delta+1}$, and for the hardcore model with fugacity $\lambda$ and bounded maximum degree $\Delta$  under the condition $\lambda\le\frac{1}{\sqrt{2}\Delta -1}$, 
we obtain dynamic sampling algorithms that upon each update of an edge or a vertex can draw a new sample within $O(1)$ incremental cost.
Meanwhile, in the ``non-uniqueness regimes'' for these models where respectively $\mathrm{e}^{-2 |\beta|} < 1- \frac{2}{\Delta}$ and $\lambda>\frac{\mathrm{e}}{\Delta-2}$, even for static and approximate sampling, either there is no local algorithm or the problem itself is intractbale.

Our dynamic sampling algorithm uses the idea of ``resampling'': 
Given a starting sample, once the graphical model changes, the part of the sample that is no longer valid is resampled (potentially multiple times).
The idea of resampling was crucial to the Moser-Tardos algorithm for constructing a satisfying solution to the Lov\'asz Local Lemma  (LLL)~\cite{moser2010constructive,haeupler2011new, Harris2013The, harvey2015algorithmic} and the algorithms for sampling uniformly distributed satisfying solution to the LLL~\cite{guo2016uniform, wilson1996generating}.
Prior to our work, it was unknown whether there is such a local resampling rule that can even generate the {\em correct distribution} for general graphical models.
One of our main conceptual contributions is to come up with an ``equilibrium'' property and show that our resampling algorithm satisfies this property at each step. 
Roughly, our equilibrium property asserts: conditioning on any subset of variable to be resampled and their current values, the remaining variables are ``consistent'' with the Gibbs distribution.
This property can easily guarantee the correctness of our sampling algorithm in a dynamic setting.
Our techniques should be of independent interest and, in particular, could be useful to extend our results to sampling from other spin systems, graphical  models over continuous distributions and/or with global constraints. 

\paragraph{Organization of this paper.}
The preliminaries are given in Section~\ref{sec:prelim}. 
In Section~\ref{sec:algorithms}, we formally define the dynamic sampling problem and give our main algorithm, the \concept{Dynamic Sampler} (Algorithm~\ref{CodeIncrLRS}). The main results are stated in Section~\ref{sec:results}, followed by the related works discussed in Section~\ref{sec:related-work}.
The the correctness and efficiency of the algorithm are analyzed in Section~\ref{sec:equilibrium} and Section~\ref{sec:convergence} respectively. 
And these are applied on specific well-studied graphical models in Section~\ref{sec:applications}. 
Finally in Section~\ref{sec:conclusion}, the conclusion and open problems are given.

\section{Preliminaries}\label{sec:prelim}
\paragraph{Graphical models.}  A (discrete) \concept{graphical model} is a tuple $\mathcal{I}=(V,E,[q],\Phi)$, where $V$ is a set of $n$ \concept{variables}, $E\subseteq 2^V$ is a set of  $m$ \concept{constraints} (or \concept{factors}), and $\Phi=(\phi_a)_{a\in E\cup V}$.
Each $v\in V$ corresponds to a variable of domain $[q]$ and is associated with a function $\phi_v:[q]\to\mathbb{R}_{\ge 0}$.
Each constraint $e\in E$ is a set of variables with $|e|>1$, and is associated with a function $\phi_e:[q]^{e}\to\mathbb{R}_{\ge 0}$.
Without loss of generality, we assume that each $\phi_v$ is {normalized} as a probability distribution over $[q]$,~i.e.~$\sum_{x\in[q]}\phi_v(x)=1$; and each $\phi_e$ is {normalized} as $\phi_e:[q]^{e}\to[0,1]$. A constraint $e$ is called a \concept{hard} constraint if $\phi_e$ is Boolean-valued, and otherwise it is called a \concept{soft} constraint.

Each \concept{configuration} $\sigma\in[q]^V$ assigns every variable one of the $q$ possible values, and is assigned following weight:
\[
w(\sigma)\triangleq \prod_{v\in V}\phi_v\left(\sigma_v\right)\prod_{e\in E}\phi_e\left(\sigma_{e}\right),
\]
where $\sigma_e$ stands for the restriction of $\sigma$ on subset $e\subseteq V$.

\paragraph{Gibbs distribution.}
The \concept{Gibbs distribution} of graphical model $\mathcal{I}$, denoted as $\mu=\mu_{\mathcal{I}}$, is defined as $$\mu(\sigma)\triangleq \frac{w(\sigma)}{Z}$$ where $$Z \triangleq \sum_{\sigma\in[q]^V}w(\sigma)$$ is the \concept{partition function}. We simply write $\mu(\sigma)\propto w(\sigma)$.
Note that normalization of the constraints as described above does not change the Gibbs distribution $\mu$.

\paragraph{Dependency graphs.}

The \concept{dependency graph} of the graphical model $\mathcal{I}$ is a graph with vertex set $E$, where any two constraints $e,f\in E$ are adjacent in the dependency graph if and only if they share a variable.
For any constraint $e\in E$, let $$\Gamma(e)\triangleq\{f\in E\setminus\{e\}\mid  f\cap e\neq\emptyset\}$$ denote the neighborhood of $e$ in the dependency graph.

\paragraph{Variable sets.}
For a subset $D\subseteq V\cup E$ of variables and constraints, we use $$\vbl{D}\triangleq (D\cap V)\cup\left(\bigcup_{e\in D\cap E}e\right)$$ to denote the set of variables in $D$ or involved in constraints in $D$. 

\paragraph{Internal and boundary constraints.}
Given a set of variables $S\subseteq V$, we use $$E(S)\triangleq\{e\in E\mid e\subseteq S\}$$ to denote the set of \concept{internal} constraints defined on variables in $S$; $$\delta(S)\triangleq\{e\in E\mid e\not\subseteq S\wedge e\cap S\neq\emptyset\}$$ the set of \concept{boundary} constraints; and $$E^+(S)\triangleq E(S)\cup\delta(S)$$ the set of constraints that use variables in $S$.

\section{Dynamic Sampling}
\label{sec:algorithms}
We first describe the dynamic sampling setup that we consider in this paper.

\subsection{The problem} 

We consider dynamical graphical model that are subject to local updates.
Let $\mathcal{I}=(V,E,[q],\Phi)$ be the input graphical model.
We consider following types of local updates:
\begin{itemize}
\item \textbf{updates for constraints:} 
modifying the functions $\phi_e$ of existing constraints $e\in E$; or adding new constraints $(e,\phi_e)$ where $e\not\in E$;
\item \textbf{updates for variables:} 
modifying the functions $\phi_v$ of variables $v\in V$.
\end{itemize}
We consider the general case where a sequence of updates may be applied to the graphical model simultaneously.
An \concept{update request}, or simply an \concept{update}, is represented as a pair $(D,\Phi_D)$. Here $D\subseteq V\cup 2^V$ contains the variables and constraints to be updated, where each $a\in D$ is either a variable $a\in V$, or an existing constraint $a\in E$, or a new constraint $a\in 2^V\setminus E$ with $|e|>1$; and $\Phi_D=(\phi_a)_{a\in D}$ specifies the function $\phi_a$ that we are updating to for all $a\in D$.
\begin{remark}[\bf Deletion and other updates]
The {deletion} of a constraint $e\in E$ can be realized by updating its function $\phi_e$ to the constant function with value 1.
The addition/deletion of independent variables with no incident constraint is trivial to implement. 
Therefore, without loss of generality we assume the variable set $V$ remains unchanged.
\end{remark}

The problem of dynamic sampling from graphical model is then defined as following:
  \par\addvspace{.5\baselineskip}
\framebox{
  \noindent
  \begin{tabularx}{14cm}{@{\hspace{\parindent}} l X c}
    \multicolumn{2}{@{\hspace{\parindent}}l}{\underline{Dynamic Sampling from Graphical Model}} \\
    \textbf{Input:} & a graphical model $\mathcal{I}$, a sample $\boldsymbol{X}\sim\mu_{\mathcal{I}}$,\\
    & and an update  $(D,\Phi_D)$ that modifies $\mathcal{I}$ to $\mathcal{I}'$;\\
    \textbf{Output:} & a sample $\boldsymbol{X}'\sim\mu_{\mathcal{I}'}$.
  \end{tabularx}
 }
\par\addvspace{.5\baselineskip}

\noindent
We assume that the update $(D,\Phi_D)$ is fixed arbitrarily by an offline adversary independently of the sample $\boldsymbol{X}\sim\mu_{\mathcal{I}}$.
A stronger adaptive adversary is discussed later in Remark~\ref{remark:adversary}.

\begin{remark}[\bf Assumption about starting sample]
The assumption of having a sample from the current graphical model can be easily achieved initially by starting from an empty graphical model on $n$ variables, whose Gibbs distribution is the product distribution $\bigotimes_{v\in V}\phi_v$, after which the availability of such sample is invariant assuming the correctness of dynamic sampling.
\end{remark}

\subsection{The dynamic sampler}
\begin{algorithm}[ht]
\SetKwInOut{Input}{Input}
\SetKwInOut{Update}{Update}
\SetKwInOut{Output}{Output}
\Input{a graphical model $\mathcal{I}$ and a random sample $\boldsymbol{X}\sim\mu_{\mathcal{I}}$;}
\Update{an update $(D,\Phi_D)$ which modifies $\mathcal{I}$ to $\mathcal{I}'$;}
\Output{a random sample $\boldsymbol{X}\sim\mu_{\mathcal{I}'}$;}
$\mathcal{R}\gets \vbl{D}$\; \label{Incr-initial-R}
\While{$\mathcal{R}\neq \emptyset$\label{Incr-while-loop-1}
}{
		$(\boldsymbol{X},\mathcal{R})\gets$\Resample($\mathcal{I}'$, $\boldsymbol{X}$, $\mathcal{R}$)\;\label{Incr-while-loop-2}
		}
\Return{$\boldsymbol{X}$}\;
\caption{Dynamic Sampler}\label{CodeIncrLRS}
\end{algorithm}

\begin{algorithm}[ht]
\SetKwInOut{Input}{Input}
\SetKwInOut{Output}{Output}
\Input{a graphical model $\mathcal{I}=(V, E, [q], \Phi)$,  a configuration $\boldsymbol{X}\in[q]^V$ and a $\mathcal{R}\subseteq V$;}
\Output{a new pair $(\boldsymbol{X}',\mathcal{R}')$ of configuration $\boldsymbol{X}'\in[q]^V$ and subset $\mathcal{R}'\subseteq V$;}
		for each $e\in E^+(\mathcal{R})$, {in parallel}, compute $\kappa_e\triangleq\frac{1}{\phi_e\left(X_{e}\right)}\min_{x \in [q]^e:\,x_{e\cap \mathcal{R}} = X_{e \cap \mathcal{R}}}\phi_e(x)$\;\label{LRS-compute-kappa}
		for each $v \in {\mathcal{R}}$, {in parallel}, resample $X_v\in[q]$ independently according to distribution $\phi_v$\; \label{LRS-resample-step}
		for each $e\in E^+(\mathcal{R})$, {in parallel}, sample $F_e\in\{0,1\}$ ind.~with $\Pr[F_e=0]=\kappa_e\cdot \phi_e\left(X_{e}\right)$\;\label{LRS-remove-step}
$\boldsymbol{X}'\gets\boldsymbol{X}$ and $\mathcal{R}'\gets\bigcup_{e\in E: F_e=1} e$\;	\label{LRS-construct-R}
\Return{$(\boldsymbol{X}',\mathcal{R}')$.}
\caption{\Resample($\mathcal{I}$, $\boldsymbol{X}$, $\mathcal{R}$)}\label{CodeResample}
\end{algorithm}

\noindent
We give a dynamic sampling algorithm (Algorithm~\ref{CodeIncrLRS}) for the above problem.
The algorithm proceeds by calling a  resampling subroutine (Algorithm~\ref{CodeResample}) to resample the variables in $\vbl{D}=(D\cap V)\cup\left(\bigcup_{e\in D\cap E}e\right)$, which is the set of all variables involved in the update $(D,\Phi_D)$.

\paragraph{The local resampling procedure.}  
The resampling subroutine, the \Resample{} (Algorithm~\ref{CodeResample}), is the core of our algorithm. 
For a fixed a graphical model $\mathcal{I}=(V,E,[q],\Phi)$, the resampling procedure takes as input a pair $(\boldsymbol{X},\mathcal{R})$ of configuration $\boldsymbol{X}\in[q]^V$ and  subset $\mathcal{R}\subseteq V$ of variables, where $\mathcal{R}$ represents the current \concept{resample set} that contains the ``problematic'' variables to be resampled.
The resampling procedure transforms this input pair $(\boldsymbol{X},\mathcal{R})$ to a random pair $(\boldsymbol{X}',\mathcal{R}')$ of new configuration $\boldsymbol{X}'\in[q]^V$ and new resample set $\mathcal{R}'\subseteq V$ by the following simple rules:
\begin{enumerate}
\item 
For each problematic variable $v\in \mathcal{R}$, resample its value $X_v\in[q]$ independently according to the distribution $\phi_v$. 
We denote by $\boldsymbol{X}'$ the configuration resulting from this resampling.
 \item 
For each constraint $e$ affected by the resampling (because some variables in $e$ are resampled), this constraint $e$ is violated  ($F_e=1$ in algorithm)  independently with probability $1-\kappa_e\cdot\phi(X_e')$. The variables involved in the violated constraints form the new resample set $\mathcal{R}'$.
\end{enumerate}
Here $\kappa_e\in[0,1]$ is a correcting factor computed from the pre-resample configuration $\boldsymbol{X}$ as:
\begin{align}
\kappa_e=\frac{1}{\phi_e\left(X_{e}\right)}\min_{\substack{x \in [q]^e\\x_{e\cap \mathcal{R}} = X_{e \cap \mathcal{R}}}}\phi_e(x) \quad (\text{with convention }0/0=1).
\label{eq:correction-factor}
\end{align}
where the min gives the  minimum value of function $\phi_e$  estimated from observing $\boldsymbol{X}$ within $\mathcal{R}$. 

Note that $\kappa_e$'s are calculated from the configuration $\boldsymbol{X}$ {before} the resampling (thus Line~\ref{LRS-compute-kappa} and Line~\ref{LRS-resample-step} in Algorithm~\ref{CodeResample} are not interchangeable).
For the \emph{internal} constraints $e\in E(\mathcal{R})$, $\kappa_e$ is always 1 thus has no effect on violating such constraints.
It may only bias the probabilities of violating the \emph{boundary} constraints $e\in\delta(\mathcal{R})$ by increasing them.
Algorithm~\ref{CodeIncrLRS} repeats the above process until the resample set $\mathcal{R}$ is empty.

While our algorithm is simple, establishing its correctness, i.e., it outputs from the right distribution $\mu_{\mathcal{I}'}$, is not. 
Certain steps in the algorithm, for instance, the definition of $\kappa_e$'s, are crucial for this purpose and become more clear from the analysis.

\begin{remark}[\bf Features of the algorithm]
Unlike the MCMC sampling, our algorithm is a Las Vegas sampler that knows when it terminates --  this is important in simulations.
Also, besides being dynamic, our sampling algorithm is parallelizable, and can be implemented as communication-efficient distributed algorithms in a distributed sampling model considered~\cite{feng2017sampling,fischer2018simple}.
\end{remark}

\begin{remark}[\bf Comparison with algorithms for constructing and sampling LLL solution]
The famous Moser-Tardos algorithm~\cite{moser2010constructive,haeupler2011new, Harris2013The, harvey2015algorithmic} for constructing LLL solution also relies on local resampling of random variables that violate constraints. 
It was observed by~\cite{harris2016new, guo2016uniform} that the Moser-Tardos  algorithm does not generate the correct distribution except for very restricted types of constraints.
This was fixed by the partial rejection sampling method~\cite{guo2016uniform} for uniform sampling LLL solution (graphical models with hard constraints), by resampling an ``unblocking'' superset of violating random variables (which in our setting corresponds to the case where $\kappa_e=1$ for all boundary constraints $e$).
A crucial difference between our algorithm and all these previous resampling-based algorithms, is that {\bf our algorithm uses both the current values of the variables and the values after the resampling in determining whether a constraint is violated}.
This seems to be a key to sample correctly from general graphical models.
\end{remark}

\section{Main Results}\label{sec:results}
\subsection{The equilibrium property}
The correctness and efficiency of our dynamic sampling algorithm rely on an equilibrium property.
In this section we present this equilibrium property that is key to our results. 
First, we introduce the some preliminaries and explain our ``conditional Gibbs property''.

Let $\mathcal{I} =(V, E, [q], \Phi)$ be a graphical model with Gibbs distribution $\mu=\mu_{\mathcal{I}}$.
Given any  $S \subseteq V$ and a \concept{boundary condition} $\tau \in [q]^{V\setminus{S}}$ such that $\Pr_{X\sim\mu}[X_{V\setminus S}=\tau]>0$, let $\mu_S^{\tau}$ denote the \concept{marginal distribution} induced by $\mu$ over $S$ conditioning on $\tau$,~i.e.,
$$\forall\sigma\in[q]^S,\quad \mu_S^{\tau}(\sigma)\triangleq \Pr_{X\sim\mu}[X_S=\sigma\mid X_{V\setminus S}=\tau].$$
We extend the definition of the marginal distribution $\mu_S^{\tau}$ to the boundary conditions $\tau$ that may locally violate hard constraints.
For any  $S \subseteq V$ and $\tau \in [q]^{V\setminus{S}}$, 
let $\mgn{\mu}_{S}^{\tau}$ be the distribution over all $\sigma\in [q]^S$ such that
\begin{align}
\mgn{\mu}_{S}^{\tau}(\sigma) \propto w_S^{\tau}(\sigma)\triangleq
\prod_{v \in S}\phi_v(\sigma_v)\prod_{e \in E\atop e\cap S \neq \emptyset}\phi_e((\sigma \cup \tau)_{e}),\label{eq:conditional-margin-dist}
\end{align}
if $\sum_{\sigma\in[q]^S}w_S^{\tau}(\sigma)>0$. 
And if $\sum_{\sigma\in[q]^S}w_S^{\tau}(\sigma)=0$, $\mgn{\mu}_{S}^{\tau}(\sigma)$ is no longer a well-defined distribution, in which case we assume $\mgn{\mu}_{S}^{\tau}(\sigma)=0$ for all $\sigma\in [q]^S$ as a convention. 
Clearly, $\mu_S^{\tau}=\mgn{\mu}_S^{\tau}$ for any feasible {boundary condition} $\tau \in [q]^{V\setminus{S}}$ with $\Pr_{X\sim\mu}[X_{V\setminus S}=\tau]>0$.

Recall that our sampling algorithm maintains a random pair $(X,\mathcal{R})$ of a configuration $X\in[q]^V$ and a ``resample set''  $\mathcal{R}\subseteq V$ of problematic variables.
In the following we consider the random pair $(X,\mathcal{S})$ where $\mathcal{S}\triangleq V\setminus\mathcal{R}$ represents the ``sanity set'' which contains non-problematic variables.

\begin{definition}[{\bf Conditional Gibbs property}]\label{DefConditionalGibbs}
A random pair $(X,\mathcal{S})\in[q]^V\times 2^V$ is said to be \concept{conditionally Gibbs} with respect to $\mathcal{I}$ if for any $S\subseteq V$ and any $\tau\in[q]^{V\setminus{S}}$ that $\Pr_{(X,\mathcal{S})}[\mathcal{S} = S\wedge X_{V\setminus{S}} = \tau] > 0$, 
it holds that $\mgn{\mu}_S^{\sigma}$ is a well-defined probability distribution over $[q]^S$ and
\begin{align}
\forall \sigma\in[q]^S, \quad \Pr[X_\mathcal{S}=\sigma\mid \mathcal{S} = S\wedge X_{V\setminus{S}} = \tau]=\mgn{\mu}_S^{\tau}(\sigma),
\label{eq:conditional-Gibbs}
\end{align}
i.e., the distribution of $X_{\mathcal{S}}$ conditioning on $\mathcal{S}=S$ and $X_{V\setminus S}=\tau$ is the same as the marginal Gibbs distribution $\mgn{\mu}_S^{\tau}$.
\end{definition}
\noindent
This definition states a key property for the random pair $(X,\mathcal{R})$ generated by the resampling procedure: 
conditioning on any possible resample set $\mathcal{R}=R$ and its configuration $X_R=\tau$, the variables in $\mathcal{S}=V\setminus\mathcal{R}$ follow the marginal Gibbs distribution $\mu$ over $\mathcal{S}$ with boundary condition~$\tau$.

For a fixed  $\mathcal{I}$, each call of \Resample{} transforms the current pair $(X,\mathcal{R})$ to a new pair $(X', \mathcal{R}')\in[q]^V\times 2^V$ as:
$(X', \mathcal{R}') \gets \Resample(\mathcal{I}, X, \mathcal{R})$.
This naturally defines a Markov chain on space $[q]^V\times 2^V$ over states $(X, \mathcal{S})$ as follows:
\begin{definition}[\bf The resampling chain]\label{def:MRes}
Let $\MC{M}_{\LRes}$  denote the Markov chain on space $[q]^V\times 2^V$ over states $(X, \mathcal{S})$ defined as follows. 
Each transition $(X,\mathcal{S})\to (X',\mathcal{S}')$ of this chain is as:
\begin{align*}
\begin{array}{rl}
(X', \mathcal{R}') &\gets\Resample(\mathcal{I}, X, V\setminus \mathcal{S});\\
\mathcal{S}' &\gets V\setminus  \mathcal{R}'.
\end{array}
\end{align*}
The chain stops when $\mathcal{S}=V$. We call this chain $\MC{M}_{\LRes}$ the \concept{resampling chain}.
\end{definition}

\noindent
The reason we define this chain using the ``sanity set'' $\mathcal{S}$ (which is the complement of the ``resample set'' $\mathcal{R}$ maintained by the resampling algorithm \Resample) instead of using $\mathcal{R}$ itself is the following:
while the resampling algorithm works by fixing the problematic variables within set $\mathcal{R}$, the analysis should focus on the distribution of non-problematic variables outside $\mathcal{R}$.

We use $\MC{M}$ to abstractly denote Markov chains $(X,\mathcal{S})$ on space $[q]^V\times 2^V$. A crucial property to guarantee such chain $\MC{M}$ always sample from the correct Gibbs distribution $\mu_{\mathcal{I}}$ when stops (i.e.~conditioning on $\mathcal{S}=V$, the sample $X$ follows $\mu_{\mathcal{I}}$) is the following equilibrium condition.

\begin{condition}[\bf The equilibrium condition]
\label{ConMCLV-stationary}
If a random pair $(X,\mathcal{S})\in[q]^V\times 2^V$ is conditionally Gibbs with respect to graphical model $\mathcal{I}$, then after one-step transition of $\MC{M}$,  the new pair $(X',\mathcal{S}')$ is also conditionally Gibbs with respect to $\mathcal{I}$.
\end{condition}
\noindent
Note that the goal of resampling is to draw a sample $X$ from the correct distribution at the end when the resample set $\mathcal{R}=\emptyset$ (i.e.~$\mathcal{S}=V$). 
This equilibrium condition promises something much stronger: at any step of resampling, even while the current set of problematic variables $\mathcal{R}$ may not be empty, the sample $X$ is always faithful to the correct distribution over the remaining variables.

If this equilibrium condition indeed holds for the chain $\MC{M}_{\LRes}$ defined above, then it implies the correctness of our dynamic sampling algorithm (Algorithm~\ref{CodeIncrLRS}).
However, this equilibrium condition can  be difficult to verify in general and, in Section~\ref{EquilibriumConditions}, we present a sufficient condition which gives a refined equilibrium condition that implies Condition~\ref{ConMCLV-stationary} and is more explicit to verify.

\subsection{The correctness of the algorithm}
By verifying the equilibrium condition on the resampling chain $\MC{M}_{\LRes}$, we show that our dynamic sampler always outputs from the correct distribution.

\begin{theorem}[\bf Correctness of the dynamic sampling algorithm]\label{ThmLRSCorrect}
Assuming the input sample $\boldsymbol{X}\sim\mu_{\mathcal{I}}$, upon termination,
Algorithm~\ref{CodeIncrLRS} returns a perfect sample $\boldsymbol{X}'\sim\mu_{\mathcal{I}'}$.
\end{theorem}

\noindent
In previous resampling-based algorithms~\cite{moser2010constructive,haeupler2011new, Harris2013The, harvey2015algorithmic,guo2016uniform}, the analysis keeps track of an infinite-size table of random variables for resampling. 
In contrast, the correctness of our dynamic sampler is due to the above equilibrium condition, 
which provides a better understanding of why the algorithm always outputs from the correct distribution even in a dynamic setting,
and also provides new information for analyzing the running time.

In fact, this theorem along with the equilibrium property, are proved for a general class of resampling-based sampling algorithms (formally stated in Section~\ref{sec:meta-algorithm}) that include our dynamic sampler as a special case.

\begin{remark}[\bf Stronger adversary]\label{remark:adversary}
The above theorem holds even when the update $(D,\Phi_D)$ is provided by an online adaptive adversary satisfying certain locality property:
the update $(D,\Phi_D)$ may be correlated with the current sample $\boldsymbol{X}\sim\mu_{\mathcal{I}}$, but conditioning on any particular $(D,\Phi_D)$ and any current assignment $X_{\vbl{D}}=\tau$, the distribution of $X_S$, where $S\triangleq V\setminus \vbl{D}$, is precisely the marginal Gibbs distribution $\mu_S^\tau$ induced by $\mu_{\mathcal{I}}$ over $S$.
We call such an adversary a \concept{locally adaptive adversary}, since it covers the natural adaptive adversaries where $D$ is constructed incrementally by observing $\boldsymbol{X}$ inside $\vbl{D}$. Such adversary is stronger than the offline adversary where $(D,\Phi_D)$ is fixed arbitrarily independent of the sample $\boldsymbol{X}\sim\mu_{\mathcal{I}}$.
\end{remark}

\subsection{The running time of the algorithm}

While our algorithm is always correct, for its running time to be efficient, some conditions on the graphical model must be satisfied.
The reason is that sampling from graphical models in general is NP-hard, which is true even for static and approximate sampling~\cite{JVV86,galanis2015inapproximability}. 
The following theorem gives a sufficient condition that guarantees that our dynamic sampling algorithm is efficient and each update takes time proportional to the size of the update.
The time complexity of the algorithm is measured by the total number of individual resamplings of variables $X_v\in[q]$ and indicators $F_e\in\{0,1\}$ made by the algorithm during its execution. 

\begin{theorem}[\bf Fast convergence of the dynamic sampling algorithm]\label{ThmLRSConv}
Let $\mathcal{I}=(V,E,[q],\Phi)$ be a graphical model and $(D,\Phi_D)$ an update to $\mathcal{I}$.
Assume that the followings hold for the updated graphical model $\mathcal{I}'=(V,E',[q],\Phi')$.
For every $e\in E'$, we have $\phi'_e:[q]^e\to[B_e,1]$ for some $0< B_e\le 1$, and there is a constant $0<\delta<1$ such that 
\begin{align}
\forall e\in E',\quad
B_e\ge \left( 1-\frac{1-\delta}{d+1}\right)^{1/2},\label{eq:converge-cond}
\end{align}
where $d\triangleq\max_{e\in E}|\Gamma(e)|$ is the maximum degree of the dependency graph of $\mathcal{I}'$.
Then Algorithm~\ref{CodeIncrLRS} terminates within $O(\log |D|)$ iterations in expectation.
Further, if $d=O(1)$ and $\max_{e\in E}|e|=O(1)$,  the total number of resamplings is bounded by $O(|D|)$ in expectation.
\end{theorem}

\noindent
For general graphical models, if only knowing the parameters $B_e$'s and $d$, one cannot expect much improvement over this sufficient condition.
For example, for the Ising model (ferromagnetic or anti-ferromagnetic, with arbitrary local fields) on graphs with maximum degree $\Delta$ with inverse temperature $\beta$, the bound~\eqref{eq:converge-cond} translates to a condition $B_e=\mathrm{e}^{-2|\beta|}>1-\frac{1}{4\Delta}+o(\frac{1}{\Delta})$; while the sampling problem (even in a static and approximate setting) for the anti-ferromagnetic Ising model in the ``non-uniqueness regime'', where $\mathrm{e}^{-2|\beta|}<1-\frac{2}{\Delta}$,  is NP-hard~\cite{galanis2016inapproximability}.

While this sufficient condition on general graphical models focuses on graphical models with soft constraints, our dynamic sampling algorithm is not restricted to such settings. For specific graphical models with hard constraints, for instance, the hardcore model, we show a regime of fast convergence that improves the previous recent result in~\cite{guo2016uniform}.

\begin{remark}[{\bf Implications to static sampling}]
Theorem~\ref{ThmLRSCorrect} and Theorem~\ref{ThmLRSConv} together provide linear-time static sampling algorithms for graphical models with  $d=O(1)$ and $\max_{e\in E}|e|=O(1)$ and satisfying~\eqref{eq:converge-cond}, where the algorithm can draw a perfect sample from the graphical model within $O(n)$ cost in expectation by adding constraints onto independent random variables. 
In contrast, for MCMC sampling there is a $\Omega(n\log n)$ lower bound on the mixing time of local Markov chains~\cite{hayes2007general}.
\end{remark}

\subsection{Dynamic sampling from the spin systems}
On spin systems, e.g.~the Ising model or the hardcore model, by exploiting the equilibrium condition (Condition~\ref{ConMCLV-stationary}) we obtain dynamic samplers with improved convergence conditions.
\paragraph{The Ising/Potts model.}
In the \concept{Ising model} on graph $G=(V,E)$, each edge $e\in E$  is associated with an \concept{inverse temperature} $\beta_e\in\mathbb{R}$.
The Gibbs distribution over all configurations $\sigma\in\{-1,+1\}^V$ is defined such that 
\[
\mu(\sigma)\propto \prod_{e=(u,v)\in E}\exp(\beta_e\sigma_u\sigma_v).
\]
The model is called ferromagnetic if all $\beta_e>0$, and anti-ferromagnetic if all $\beta_e<0$.

For the Ising model, our dynamic sampler is efficient in a regime described as follows.
\begin{theorem}
\label{ThmIsing}
Assuming $\mathrm{e}^{-2|\beta_e|}\geq 1-\frac{1}{\alpha\Delta+1}$ for all $e\in E$, where $\alpha \approx 2.221\ldots$ is the root of $ \alpha = 1+\frac{2}{1+ \exp\left(-{1}/{\alpha}\right)} $, there is a dynamic sampling algorithm for the Ising model on graphs with maximum degree $\Delta=O(1)$, that upon updates of $k$ edges and vertices, returns a perfect sample within $O(\log k)$ rounds and $O(k)$ incremental costs in expectation.
\end{theorem}

\noindent
This gives the first fully dynamic Ising sampler that can handle both edge and vertex updates (addition, deletion, change of functions). 
The algorithm is parallel and in the above regime terminates within $O(\log n)$ rounds with high probability.

This bound is asymptotically tight.
In the ``non-uniqueness regime'' where $\mathrm{e}^{-2|\beta|}<1-\frac{2}{\Delta}$, for the anti-ferromagnetic Ising model, even static and approximate sampling is intractable~\cite{galanis2016inapproximability}; and for the ferromagnetic Ising model, by an argument as in~\cite{feng2017sampling} there cannot exist such local and parallel sampling algorithms (even for static and approximate sampling) due to the reconstructibility of the ferromagnetic Ising model in the non-uniqueness regime on locally tree-like graphs~\cite{dembo2010ising, gerschcnfeld7reconstruction}.

The Ising model is a major subject for static sampling. 
Previously, the famous Jerrum-Sinclair chain on even subgraphs~\cite{fill2000randomness} gives a poly-time approximate sampler for the ferro-Ising model. The same can also be obtained by a recent result of Guo and Jerrum~\cite{guo2018random} for the rapid mixing of the random cluster model, which combined with the \concept{coupling from the past} (CFTP) of Propp and Wilson~\cite{propp1996exact} also gives a poly-time perfect sampler for the ferro-Ising model. These Ising samplers work in the entire ferromagnetic regime, but require global translations of the probability space and has large polynomial running times.
For local algorithms, the rapid mixing result of Mossel and Sly~\cite{mossel2013exact} for the \concept{Glauber dynamics} (a local Markov chain, also known as \concept{heat bath} or \concept{Gibbs sampler}) gives a $O(n\log n)$-time static and approximate sampler for the ferro-Ising model  in the uniqueness regime where $\beta>0$ and $\mathrm{e}^{-2\beta}>1-\frac{2}{\Delta}$.
For Las Vegas samplers, local algorithms such as the random recycler of Fill and Huber~\cite{fill2000randomness} and the bounding chain of Huber~\cite{huber2004perfect} give linear- or near-linear time Las Vegas perfect samplers for the (ferro- or anti-ferro-) Ising model in regimes with the form $\mathrm{e}^{-2|\beta|}>1-O(\frac{1}{\Delta})$.
All these algorithms are non-parallel and none of them can deal with fully dynamic updates of edges and vertices (addition, deletion, change of functions).

The \concept{Potts model} is a generalization of the Ising model to non-Boolean states. Each instance of the Potts model is the same as Ising model. The Gibbs distribution is now defined over all configurations $\sigma\in[q]^V$, where $q\ge 2$ gives the number of spin states, such that
\[
\mu(\sigma)\propto\prod_{e=(u,v)\in E}\exp(\beta_e \cdot (2\delta(\sigma_u,\sigma_v)-1)),
\]
where $\delta(\cdot,\cdot)$ is the Kronecker delta.

For the Potts models, the same bounds as in Theorem~\ref{ThmIsing} hold.

\paragraph{The hardcore model.}
In the \concept{hardcore model} on graph $G=(V,E)$, each vertex $v\in V$ is associated with a \concept{fugacity} $\lambda_v>0$.
The Gibbs distribution is defined over all independent sets $I$ of graph $G$ as $\mu(I)\propto \prod_{v\in I}\lambda_v$. We call a $v\in I$ \concept{occupied} and a $v\not\in I$ \concept{unoccupied}.

\begin{theorem}
\label{ThmHardcore}
Assuming  $\lambda_v\leq\frac{1}{\sqrt{2}\Delta-1}$ for all $v\in V$, 
there is a dynamic sampling algorithm for the hardcore model on graphs with maximum degree $\Delta=O(1)$, that upon updates of $k$ edges and vertices, returns a perfect sample within $O(\log k)$ rounds and $O(k)$ incremental costs in expectation.
\end{theorem}

\noindent
This gives the first fully dynamic hardcore sampler that can handle both edge and vertex updates (addition, deletion, change of functions).
The algorithm is also parallel and in the above regime terminates within $O(\log n)$ rounds with high probability.

Our bound for the hardcore model is also asymptotically tight.
There is a critical threshold $\lambda_c(\Delta)=\frac{(\Delta-1)^{\Delta-1}}{(\Delta-2)^\Delta}\approx\frac{\mathrm{e}}{\Delta-2}$ known as the ``uniqueness threshold'' such that when $\lambda>\lambda_c(\Delta)$ even static and approximate sampling is intractable~\cite{galanis2016inapproximability}.

Previously, the Glauber dynamics is known to be rapidly mixing for the hardcore model with $\lambda<\lambda_c(\Delta)$ on amenable graphs~\cite{goldberg2005strong, weitz2006counting} as well as graphs with large girth and degree~\cite{efthymiou2016convergence}, and also with $\lambda \leq \frac{2}{\Delta - 2}$~\cite{vigoda1999fast, dyer2000markov} on general graphs. 
And the perfect sampling methods of~\cite{fill2000randomness, huber2004perfect} give $\tilde{O}(n)$-time perfect samplers also when $\lambda \leq \frac{2}{\Delta - 2}$.
All these algorithms are non-parallel and none of them can deal with fully dynamic updates of edges and vertices (addition, deletion, change of functions).

To achieve better convergence, our algorithm deviates slightly from Algorithm~\ref{CodeIncrLRS}, but still falls into its generalization (formally introduced in Section~\ref{sec:meta-algorithm}).
We actually show that a natural dynamic variant of the algorithm in~\cite{guo2016uniform} is always correct dynamically and improve their regime of fast convergence from $\lambda\le\frac{1}{2\sqrt{\mathrm{e}}\Delta - 1}$ to $\lambda\le\frac{1}{\sqrt{2}\Delta - 1}$.

\paragraph{Coloring.} Our algorithm is inefficient on graphical models defined by ``truly repulsive'' hard constraints, e.g.~uniform proper $q$-coloring. 
Formally, being ``truly repulsive'' means that for every constraint $e\in E$, any partial assignment $\sigma_{s}\in[q]^s$, where $s\subset e$, can be extended to a violating $\sigma_e\in[q]^e$ with $\phi_e(\sigma_e)=0$. 
For such graphical models, our algorithm is inefficient because to get a correct sample the algorithm is eventually forced to resample all variables simultaneously.
How to overcome this is left as a major open problem.

\section{Related Work}\label{sec:related-work}
The theory of Markov chain Monte Carlo (MCMC) sampling has been extensively studied in computer science, probability theory and statistics (see~\cite{levin2017markov,jerrum2003counting}). There is a substantial body of works on MCMC sampling from various graphical models, e.g.~the hardcore model~\cite{vigoda1999fast, dyer2000markov, efthymiou2016convergence}, the Ising model~\cite{jerrum1993polynomial, mossel2013exact}, and proper $q$-coloring~\cite{jerrum1995very, salas1997absence, vigoda2000improved}.

Less were known for Las Vegas perfect samplers. Some major results include the coupling from the past (CPTP) method of Propp and Wilson~\cite{propp1996exact}, Wilson's cycle-popping algorithm for uniform spanning tree~\cite{wilson1996generating}, Fill's algorithm~\cite{fill1997interruptible,fill2000extension}, the random recycler method of Fill and Huber~\cite{fill2000randomness}, Huber's bounding chain method~\cite{huber2004perfect}, and most recently the partial rejection sampling method of Guo, Jerrum, and Liu~\cite{guo2016uniform}.
See the monograph of Huber for a survey~\cite{huber2016perfect}.

The idea of resampling was used in the famous Moser-Tardos algorithm for constructing a satisfying solution to the Lov\'asz local lemma (LLL)~\cite{moser2010constructive}, followed by a line of remarkable works~\cite{haeupler2011new,Harris2013The,harvey2015algorithmic}. 
There is a profound connection between LLL and counting~\cite{moitra2017approximate, harvey2018computing}. One would expect to use the idea of resampling for sampling.
However, as observed in~\cite{harris2016new,guo2016uniform}, the Moser-Tardos algorithm does not generate uniformly distributed LLL solutions. This was fixed by the partial rejection sampling method of Guo, Jerrum and Liu~\cite{guo2016uniform} which can generate uniformly distributed LLL solutions by resampling a proper ``unblocking'' superset of violating variables. Retrospectively, Wilson's cycle-popping algorithm~\cite{wilson1996generating,guo2018tight} can be interpreted as using this method on spanning trees. Most recently, in a major breakthrough of Guo and Jerrum~\cite{guo2018simple}, a long-standing open problem in the area of approximate counting, the network reliability problem, was solved by using this method.

For sampling from a dataset, instead of from an exponential-sized space of configurations, sampling from a dynamically increasing dataset (or data stream) with small maintenance cost is a fundamental problem and is the main purpose of the classical \concept{reservoir sampling} methods~\cite{Knuth:1997:ACP:270146}.

The problem of dynamic sampling from graphical models can also be loosely seen as the sampling variant of the dynamic graph problem.
In the dynamic graph problem,  edges (constraints) are added or removed over time. 
The goal is to maintain with small  cost for each update,  a feasible solution or a locally/globally optimal solution (instead of a random solution as in sampling),  to some constraint satisfaction graph problem.
The dynamic graph problem has a rich history and is one of the major topics for algorithms and data structures. See~\cite{demetrescu2009dynamic} for a survey.

\section{Proof of Equilibrium and Correctness}\label{sec:equilibrium}
Our resampling procedure \Resample{} (Algorithm~\ref{CodeResample}) maintains a pair $(X,\mathcal{R})$ where $X\in[q]^V$ is a configuration and $\mathcal{R}\subseteq V$ is the ``resample set'' which contains the problematic variables to be resampled.
Fixed an instance $\mathcal{I}$ of graphical model, each calling of \Resample{} transforms the current pair $(X,\mathcal{R})$ to a new pair $(X', \mathcal{R}')\in[q]^V\times 2^V$ as:
\begin{align}
(X', \mathcal{R}') \gets \Resample(\mathcal{I}, X, \mathcal{R}).\label{eq:reample-set-chain}
\end{align}
Recall the Markov chain $\MC{M}_{\LRes}$ on space $[q]^V\times 2^V$, called the \concept{resampling chain}, which is defined in Definition~\ref{def:MRes}. 
The chain $\MC{M}_{\LRes}$ is defined over states $(X, \mathcal{S})$, where $\mathcal{S}\triangleq V\setminus \mathcal{R}$ stands for the ``sanity set'' that contains the non-problematic variables. 
Each transition $(X,\mathcal{S})\to (X',\mathcal{S}')$ of chain $\MC{M}_{\LRes}$ is derived from the above transition~\eqref{eq:reample-set-chain} with $\mathcal{S}= V\setminus \mathcal{R}$ and $\mathcal{S}'= V\setminus \mathcal{R}'$.

\subsection{Equilibrium conditions}
\label{EquilibriumConditions}
Recall the definition of the \concept{conditional Gibbs property} of a random pair $(X, \mathcal{S})\in [q]^V\times 2^V$, defined in Definition~\ref{DefConditionalGibbs}, which basically says that conditioning on any fixed set $\mathcal{S}=S$ and any boundary condition $X_{V\setminus S}=\tau$ specified on variables in $V\setminus S$, the distribution of $X_{S}$ is the same as the marginal Gibbs distribution $\mgn{\mu}_S^{\tau}$ on $S$ with boundary condition $\tau$ (defined as~\eqref{eq:conditional-margin-dist}).
Note that for the random $(X, \mathcal{S})$ with this property, in particular, conditioning on $\mathcal{S}=V$ the sample $X$ follows the correct Gibbs distribution $\mu_{\mathcal{I}}$.

We use $\MC{M}$ to abstractly denote Markov chains $(X,\mathcal{S})$ on space $[q]^V\times 2^V$. 
Recall the equilibrium condition stated in Condition~\ref{ConMCLV-stationary}, which says that the property of $(X,\mathcal{S})$ being conditionally Gibbs with respect to a graphical model $\mathcal{I}$ is invariant under transitions of the chain $\MC{M}$.

\paragraph{Implication to the correctness of dynamic sampling.}
Note that if Condition~\ref{ConMCLV-stationary} indeed holds for the chain $\MC{M}_{\LRes}$, then the correctness of our dynamic sampling algorithm (Algorithm~\ref{CodeIncrLRS}) stated in Theorem~\ref{ThmLRSCorrect} follows directly.
This is because for a $(X,\mathcal{S})$ with the conditional Gibbs property, whenever the algorithm stops (when $\mathcal{S}=V$), the sample $X$ follows the current correct Gibbs distribution. And with Condition~\ref{ConMCLV-stationary}, the conditional Gibbs property is invariant under resampling \Resample{}. 

It only remains to verify that the conditional Gibbs property is also invariant against updates, which is easy by the following argument.
Algorithm~\ref{CodeIncrLRS} starts with a sample $X\sim\mu_{\mathcal{I}}$ from the current graphical model $\mathcal{I}$.
And upon update $(D,\Phi_D)$ that changes $\mathcal{I}$ to a new instance $\mathcal{I}'$, in Algorithm~\ref{CodeIncrLRS} we start running the chain $\MC{M}_{\LRes}$ with the initial state $(X,V\setminus\vbl{D})$, which is obviously conditionally Gibbs with respect to $\mathcal{I}$ as $X\sim\mu_{\mathcal{I}}$ and $\vbl{D}$ is independent of $X$.
Consequently, this $(X,V\setminus\vbl{D})$ must also be conditionally Gibbs with respect to the new instance $\mathcal{I}'$, because $\mathcal{I}$ and $\mathcal{I}'$ differ only at functions $\phi_v$ and $\phi_e$ for $v,e\in D$, whose definitions do not affect whether $(X,V\setminus\vbl{D})$ is conditionally Gibbs with respect to $\mathcal{I}$ (or $\mathcal{I}'$).

Note that the above argument remains valid even when the update $(D,\Phi_D)$ is provided by a locally adaptive adversary as described in Remark~\ref{remark:adversary}, because with such adversary it still holds that the initial state $(X,V\setminus\vbl{D})$ of the chain $\MC{M}_{\LRes}$ is conditionally Gibbs with respect to $\mathcal{I}$, and the rest follows.

\paragraph{A refined equilibrium condition.}
In general, the equilibrium condition stated in Condition~\ref{ConMCLV-stationary} can still be difficult to verify.
Here we give a sufficient condition which implies Condition~\ref{ConMCLV-stationary} and is more explicit to verify.
Recall that $\mgn{\mu}_S^{\tau}$ represents the marginal Gibbs distribution on $S$ with boundary condition $\tau$, which is formally defined in~\eqref{eq:conditional-margin-dist}.

Our refined equilibrium condition is stated as follows.

\begin{condition}[\bf Refined equilibrium condition]
\label{ConMCLV-DetailedBalance}
Let $\mathcal{I} =(V, E, [q], \Phi)$ be a graphical model. 
Let $\MC{M}$ be a Markov chain on space $[q]^V \times 2^V$ with 
transition matrix $P$.
For any tuple $(S, \sigma, T, \tau)$ where $S, T \subseteq V$, $\sigma\in [q]^{V\setminus{S}}$ and $\tau \in [q]^{V\setminus{T}}$, it holds that
\begin{align*}
\forall \boldsymbol{y} \in [q]^V\text{ where }y_{V\setminus{T}}=\tau:\qquad \sum_{\boldsymbol{x} \in[q]^V \atop x_{V\setminus{S}}=\sigma }\mgn{\mu}_{S}^{\sigma}(x_S)\cdot P((x, S), (y, T)) = \mgn{\mu}_{T}^\tau(y_T)\cdot C(S, \sigma, T, \tau),
\end{align*}
where $C(S, \sigma, T, \tau)\geq 0$ is a finite constant which depends only on $(S, \sigma, T, \tau)$.
\end{condition}

\noindent
This condition describes a linear system with certain consistency requirement.
The equations in the system are grouped according to the tuples $(S, \sigma, T, \tau)$, where all equations in the same group corresponding to a $(S, \sigma, T, \tau)$ involve only those $x,y\in[q]^V$ with $x_{V\setminus S}=\sigma$ and $y_{V\setminus T}=\tau$ and has the same value on $C(S, \sigma, T, \tau)$. 
The correct resampling algorithm specified by $P$ is in fact a solution to this system.

Intuitively, this condition guarantees that the equilibrium in Condition~\ref{ConMCLV-stationary} holds in the following refined sense. 
Given a random $(X,\mathcal{S})$ that is conditionally Gibbs, conditioning on any possible $\mathcal{S}=S$ and $X_{V\setminus S}=\sigma$, after one-step transition of $\MC{M}$, the resulting state $(Y,\mathcal{T})$ (starting from the $(X,\mathcal{S})$ generated with the fixed $\mathcal{S}=S$ and $X_{V\setminus S}=\sigma$) is still conditionally Gibbs.

It is then easy to verify that Condition~\ref{ConMCLV-stationary} is implied by this refined equilibrium condition.

\begin{lemma}[\bf Sufficiency of Condition~\ref{ConMCLV-DetailedBalance} to  Condition~\ref{ConMCLV-stationary}]\label{LemRefineToMacro}
If a  Markov chain $\MC{M}$ on space $[q]^V \times 2^V$ satisfies Condition~\ref{ConMCLV-DetailedBalance}, then it satisfies Condition~\ref{ConMCLV-stationary}.
\end{lemma}

\begin{proof}
Let $(\boldsymbol{X}, \mathcal{S})\in[q]^V\times2^V$ be a random pair.
Consider a one-step transition of the chain $\MC{M}$ from $(\boldsymbol{X}, \mathcal{S})$ to $(\boldsymbol{Y}, \mathcal{T})$. 
We define
\begin{align*}
H \triangleq \left\{(S, \sigma)\mid S\subseteq V,\sigma \in [q]^{V\setminus{S}}, \Pr[X_{V\setminus{S}} = \sigma\wedge \mathcal{S}= S] > 0  \right\}.	
\end{align*}
Assume that $(\boldsymbol{X}, \mathcal{S})$ is conditionally Gibbs with respect to $\mathcal{I}$, which means that for any $(S, \sigma) \in H$, the marginal Gibbs distribution $\mgn{\mu}_S^\sigma$ is a well-defined probability distribution over $[q]^S$ and is also the same as the distribution of $X_{\mathcal{S}}$ conditioning on $\mathcal{S}=S$ and $X_{V\setminus{S}}=\sigma$. 
Then for any $T\subseteq V$, $\tau \in [q]^{V\setminus{T}}$, and any ${y} \in [q]^V$ that $y_{V\setminus{T}}=\tau$, we have
\begin{align*}
 \Pr[\boldsymbol{Y}=y\wedge\mathcal{T}=T]&=\sum_{(S,\sigma)\in H}\Pr[X_{V\setminus{S}}=\sigma\wedge \mathcal{S}=S]\sum_{\boldsymbol{x} \in[q]^V \atop x_{V\setminus{S}}=\sigma }\mgn{\mu}_{S}^{\sigma}(x_S)\cdot P((x, S), (y, T))\\
 &=\mgn{\mu}_{T}^\tau(y_T)\sum_{(S, \sigma) \in H}C(S, \sigma, T, \tau)\cdot \Pr[X_{V\setminus{S}}=\sigma\wedge\mathcal{S}=S]\qquad \text{(Condition~\ref{ConMCLV-DetailedBalance})} \\
 &=C'\cdot\mgn{\mu}_{T}^\tau(y_T),
\end{align*}
for some $C'=C'(T, \tau)$ which does not depend on $y_T$.
Therefore, fixed any $T\subseteq V$ and any $\tau \in [q]^{V\setminus{T}}$ which is possible for $(\boldsymbol{Y}, \mathcal{T})$ (so that $C'(T,\tau)>0$), the probability $\Pr[Y_{\mathcal{T}}=\cdot\wedge \mathcal{T}=T\wedge Y_{V\setminus{T}}=\tau]$ is proportional to $\mgn{\mu}_{T}^\tau(\cdot)$. This shows that the $Y_{\mathcal{T}}$ conditioning on any possible $\mathcal{T}=T$ and $y_{V\setminus{\mathcal{T}}}=\tau$ follows distribution $\mgn{\mu}_{T}^\tau$, which means $(\boldsymbol{Y}, \mathcal{T})$ is conditionally Gibbs with respect to $\mathcal{I}$.
\end{proof}

\subsection{Equilibrium property of resampling algorithms}
We then verify the refined equilibrium condition (Condition~\ref{ConMCLV-DetailedBalance}) on the resampling chain $\MC{M}_{\LRes}$.
By Lemma~\ref{LemRefineToMacro}, this shows that $\MC{M}_{\LRes}$ also satisfies Condition~\ref{ConMCLV-stationary}, which as discussed, implies the correctness of our dynamic sampler (Algorithm~\ref{CodeIncrLRS}), as stated in Theorem~\ref{ThmLRSCorrect}.
\begin{lemma}
\label{LemMRDetailedBalance}
The Markov chain $\MLRes$ satisfies Condition~\ref{ConMCLV-DetailedBalance}.
\end{lemma}
\begin{proof}
Let $P$ be the transition matrix of $\MLRes$. Fix any tuple $(S, \sigma, T,\tau)$, where $S, T \subseteq V$, $\sigma\in [q]^{V\setminus{S}}$ and $\tau \in [q]^{V\setminus{T}}$, and any $y\in[q]^V$ that $y_{V\setminus{T}}=\tau$.
Condition~\ref{ConMCLV-DetailedBalance} holds if the following always holds
\begin{align}
\label{eq-final-target}
\sum_{\boldsymbol{x} \in[q]^V \atop x_{V\setminus{S}}=\sigma }\mgn{\mu}_{S}^{\sigma}(x_S)\cdot P((x, S), (y, T)) = \mgn{\mu}_{T}^\tau(y_T)\cdot C(S, \sigma, T, \tau),
\end{align}	
where $C(S, \sigma, T, \tau)\ge 0$ depends only on $(S, \sigma, T, \tau)$.

Observe that Algorithm~\ref{CodeResample} only resamples the variables in $V\setminus{S}$, thus we have
\begin{align*}
\forall x\in [q]^V: \qquad  x_S \neq y_S  \quad\Longrightarrow\quad 	P((x, S), (y, T)) = 0.
\end{align*}
Therefore, on the LHS of equation~\eqref{eq-final-target}, there is only one $x\in[q]^V$ with $x_{V\setminus S}=\sigma$ that may have  non-zero $P((x, S), (y, T))$. This $x=x(S,\sigma, y)\in[q]^V$ is uniquely determined as 
\begin{align}
\label{eq:definition-x}
x_v=\begin{cases}
\sigma_v & v\in V\setminus S, \\
y_v & v\in S. 
\end{cases}
\end{align}
We simply refer to this $x(S,\sigma, y)$ as $x$ when $S,\sigma$ and $y$ are clear in the context.

The condition in~\eqref{eq-final-target} is then simplified to
\begin{align}
\label{eq-target}
\mgn{\mu}_{S}^{\sigma}(x_S)\cdot P((x, S), (y, T)) = \mgn{\mu}_{T}^\tau(y_T)\cdot C(S, T,\sigma, \tau),
\end{align}
where $x$ is as constructed in~\eqref{eq:definition-x}.

We first calculate the $\mgn{\mu}_{S}^{\sigma}(x_S)$ that appears in the LHS of the above equation.
\begin{claim}
\label{ClaimMu}
The following holds for the $\mgn{\mu}_{S}^{\sigma}(x_S)$ in the LHS of~\eqref{eq-target}:	
\begin{align*}
\mgn{\mu}_{S}^{\sigma}(x_S) = C_1\prod_{v \in S \cap T}\phi_v(y_v) \prod_{e \in \delta(S) \cap E^+(T) }\phi_e(x_e) \prod_{e \in E(S) \cap E^+(T) }\phi_e(y_e),
\end{align*}
where $C_1 = C_1(S, \sigma, T, \tau) \geq 0$ depends only on $S, \sigma, T, \tau$.  
\end{claim}

We then observe that we can make the following assumption without loss of generality: 
\begin{align}\label{eq-situation}
\forall e\in E^+(T)\cap \delta(V\setminus S),\quad
\forall z\in[q]^e \text{ that } z_{e\setminus S}=x_{e\setminus S},\qquad 
\phi_e(z)>0.
\end{align}
If otherwise there exists a $e\in E^+(T)\cap \delta(V\setminus S)$ such that $\phi_e(y_e)=0$ for some $y_e\in[q]^e$ with $y_{e\setminus S}=x_{e\setminus S}$, then the LHS of equation~\eqref{eq-target} must have value 0, thus~\eqref{eq-target} holds trivially with $C(S, T,\sigma, \tau)=0$.
This can be verified by the following two cases. Case.1: $\phi_e\left(x_{e}\right) = 0$, where $x$ is as constructed in~\eqref{eq:definition-x}. Then immediately we have $\mgn{\mu}_{S}^{\sigma}(x_S) = 0$. Case.2: $\phi_e\left(x_{e}\right) >0$. Then in Algorithm~\ref{CodeResample} with input $(x,V\setminus S)$ we must have $\kappa_e=\min_{z\in[q]^V:z_{e\setminus S}=x_{e\setminus S}}\phi_e(z)=0$, which means $e$ must be contained in the new resample set, however we have $e\in E^+(T)$, which means $P((x,S),(y,T))=0$.

We then calculate the transition probability $P((x, S), (y, T))$ under assumption~\eqref{eq-situation}.

\begin{claim}
\label{ClaimP}
The following holds for the $P((x,S), (y, T))$ in the LHS of~\eqref{eq-target}:
\begin{align*}
P((x, S), (y, T)) = C_2\prod_{v \in T \setminus S}	\phi_v(y_v)\prod_{e \in \delta(S) \cap E^+(T ) }\frac{\phi_e(y_e)}{\phi_e(x_e)}\prod_{e \in E(V \setminus S) \cap E^+(T) }\phi_e(y_e),
\end{align*}
where $C_2 = C_2(S, \sigma, T, \tau) \geq 0$ depends only on $(S, \sigma, T, \tau)$. 
\end{claim}
\noindent
The product in above claim is well-defined and has finite value because the ratio $\frac{\phi_e(y_e)}{\phi_e(x_e)}$ has finite value assuming~\eqref{eq-situation}.

Combining~Claim~\ref{ClaimMu} and Claim~\ref{ClaimP}, we have
\begin{align*}
\mgn{\mu}_{S}^{\sigma}(x_S)\cdot P((x, S), (y, T)) 
&= C_1C_2\prod_{v \in T} \phi_v(y_v) \prod_{e \in E^+(T)}\phi_e(y_e)\\
&=C_1C_2\prod_{v \in T}\phi_v(y_v)\prod_{e \in E^+(T)}\phi_e((y_T \cup \tau)_{e}) &(\text{since }y_{V\setminus T}=\tau),
\end{align*}
which is precisely $\mgn{\mu}_{T}^\tau(y_T)\cdot C(S, \sigma, T, \tau)$ for some $C(S, \sigma, T, \tau)$ that depends only on $S, \sigma, T, \tau$. The first equation is due to $\delta(S), E(S)$ and $E(V \setminus S)$ are disjoint and $\delta(S) \cup E(S) \cup E(V \setminus S) = E$.

\end{proof}

We then prove Claim~\ref{ClaimMu} and Claim~\ref{ClaimP}.

\begin{proof}[Proof of Claim~\ref{ClaimMu}]
Recall that $y_{V \setminus T} = \tau$ and $x$ is constructed as that $x_{V \setminus S} = \sigma$ and $x_S=  y_{S}$.
By the definition of the marginal Gibbs distribution, we have
\begin{align}
\label{eq-ClaimMU-1}
\mgn{\mu}_S^\sigma(x_S)&= C_{1,1} \prod_{v \in S}\phi_v(x_v)\prod_{e \in E^+(S)}\phi_e(x_e),
\end{align}
where $C_{1,1}$ is the reciprocal of the partition function for $\mgn{\mu}_{S}^{\sigma}$, which depends only on $S, \sigma$.

Note that $S$ can be partitioned into two disjoint subsets $S\setminus T$ and $S\cap T$. Moreover, due to the construction of $x$ and $y$, we have $x_v=y_v=\tau_v$ for every $v\in S\setminus T$, and $x_v=y_v$ for every $v\in S\cap T$. 
Therefore,
\begin{align*}
\prod_{v \in S}\phi_v(x_v) = C_{1,2}\prod_{v \in S \cap T} \phi(y_v),
\end{align*}
where $C_{1,2} = \prod_{v \in S \setminus T}\phi_v(\tau_v)$ depends only on $S, T$ and  $\tau$.

Again, note that $E^+(S)$ can be partitioned into two disjoint subsets $E^+(S)\cap E(V\setminus T)$ and $E^+(S)\cap E^+(T)$ because $E(V\setminus T)$ and $E^+(T)$ are complement to each other. Moreover, for every $e\in E(V\setminus T)$, we have $x_e$ fully determined by $(S,\sigma,T,\tau)$, because $x_{V \setminus S} = \sigma$ and $x_{S \setminus  T} = \tau_{S \setminus T}$. Therefore,
\begin{align*}
\prod_{e \in E^+(S)}\phi_e(x_e) &= C_{1,3}\prod_{e \in E^+(S) \cap E^+(T)}\phi_e(x_e),
\end{align*}
where $C_{1,3} = \prod_{e \in E^+(S) \cap E(V \setminus T)}\phi_e(x_e)$ depends only on $(S, \sigma, T, \tau)$, in particular, $x_v=\sigma_v$ if $v\in V\setminus S$ and $x_v=\tau_v$ if $v\in S\setminus T$.

And the set $E^+(S) \cap E^+(T)$ can be further partitioned into two disjoint subsets $\delta(S)\cap E^+(T)$ and  $E(S)\cap E^+(T)$ because $E^+(S)=\delta(S)\uplus E(S)$. Moreover, we have $x_e=y_e$ for every $e\in E(S)\cap E^+(T)$ because $x_S=y_S$. Therefore,
\begin{align*}
\prod_{e \in E^+(S) \cap E^+(T)}\phi_e(x_e)=\prod_{e \in \delta(S) \cap E^+(T)}\phi_e(x_e)\prod_{e \in E(S) \cap E^+(T)}\phi_e(y_e).
\end{align*}

Combining everything together, we can rewrite~\eqref{eq-ClaimMU-1} as:
\begin{align*}
\mgn{\mu}_S^\sigma(x_S)&= C_1\prod_{v \in S \cap T} \phi(y_v)\prod_{e \in \delta(S) \cap E^+(T)}\phi_e(x_e)\prod_{e \in E(S) \cap E^+(T)}\phi_e(y_e), 
\end{align*}
where $C_1=C_{1,1}C_{1,2}C_{1,3}$ depends only on $(S,\sigma,T,\tau)$,
which proves the claim.
\end{proof}

\begin{proof}[Proof of Claim~\ref{ClaimP}]
Fix a tuple $(S,\sigma,T,\tau)$ and $y\in[q]^V$ satisfying $y_{V\setminus T}=\tau$. Let $x\in[q]^V$ be constructed as~\eqref{eq:definition-x}. We then calculate the transition probability $P((x,S),(y,T))$.

Recall that the definition of the chain $\MC{M}_\LRes$ in Definition~\ref{def:MRes}. Algorithm~\ref{CodeResample} takes $(x,R)$ as input where $R\triangleq V \setminus S$. For each $v\in R$ a random value $Y_v\in[q]$ is sampled independently according to the distribution $\phi_v$, and for each $v\in S=V\setminus R$, we simply assume $Y_v=x_v$. For each $e\in E^+(R)$, a random value $F_e\in\{0,1\}$ is sampled independently with $\Pr[F_e=0]=\kappa_e\phi_e(Y_e)$ where $\kappa_e\triangleq \min_{z\in[q]^e:z_{e\setminus S}=x_{e\setminus S}}\phi_e(z)/\phi_e(x_e)$ (with convention $0/0=1$). Finally a random set $\mathcal{R}'\subseteq V$ is constructed as $\mathcal{R}'=\bigcup_{e \in E^+(R) \atop F_e= 1}e$.

$P((x,S),(y,T))$ is the probability that $Y=y$ and $\mathcal{R}'=R'$ where $R'\triangleq V\setminus T$, which occurs if and only if following  events occur simultaneously.
\begin{align*}
\mathcal{A}_1:&\quad Y_R=y_R;\\
\mathcal{A}_2:&\quad \exists \mathcal{F} \subseteq E^+(R)\cap E(R'),\text{ s.t. }\vbl{\mathcal{F}}=R'\wedge(\forall e\in\mathcal{F}, F_e=1);\\
\mathcal{A}_3:&\quad \forall e\in E^+(R)\setminus E(R'), F_e=0.
\end{align*}
The first event guarantees that $Y=y$, and the other two events together guarantee that $\mathcal{R}'=R'$. Therefore, by chain rule
\[
P((x,S),(y,T))=\Pr[\mathcal{A}_1\wedge\mathcal{A}_2\wedge\mathcal{A}_3]=\Pr[\mathcal{A}_1]\cdot \Pr[\mathcal{A}_2\mid \mathcal{A}_1]\cdot\Pr[\mathcal{A}_3\mid \mathcal{A}_1\wedge\mathcal{A}_2 ].
\]
We then calculate these probabilities separately.
\begin{itemize}
\item Since each $Y_v$ for $v\in R$ is sampled independently according to distribution $\phi_v$, we have
\begin{align}
\label{eq-variable-transform}
\Pr[\mathcal{A}_1] 
&=\Pr[Y_R=y_R]\notag\\
&= \prod_{v\in R}\phi_v(y_v)\notag\\
&=\prod_{v \in R\cap R'}\phi_v(\tau_v)\prod_{v \in R\setminus R' }\phi_v(y_v) & (y_{R'}=y_{V\setminus T}=\tau)\notag\\
&=C_{2,1} \prod_{v \in T \setminus S}\phi_v(y_v), & (R=V\setminus S\text{ and }R'= V\setminus T)
\end{align}
where $C_{2,1} = C_{2,1}(S, T, \tau) = \prod_{v \in V \setminus (T \cup S)}\phi_v(\tau_v)$ depends only on $S, T$ and $\tau$.
\item Event $\mathcal{A}_2$ can be expressed as the union of a collection of disjoint events, where each disjoint event corresponds to a $\mathcal{F} \subseteq E^+(R)\cap E(R')$ with $\vbl{\mathcal{F}}=R'$ and occurs when $\mathcal{F}$ gives the precise set of $e$'s with $F_e=1$.
Then, we have
\begin{align*}
\Pr[\mathcal{A}_2 \mid \mathcal{A}_1]
&=\sum_{\mathcal{F} \subseteq E^+(R)\cap E(R') \atop \vbl{\mathcal{F}}=R'}\Pr\left[\,\forall e\in\mathcal{F}, F_e=1\wedge \forall e\in (E^+(R)\cap E(R'))\setminus \mathcal{F}, F_e=0\mid \mathcal{A}_1\,\right]\\
&=\sum_{\mathcal{F} \subseteq E^+(R)\cap E(R') \atop \vbl{\mathcal{F}}=R'}\left( \prod_{e \in \mathcal{F}} \Pr[F_e = 1 \mid \mathcal{A}_1] \prod_{e \in E^+(R)\cap E(R') \atop e\not\in \mathcal{F}}\Pr[F_e = 0 \mid \mathcal{A}_1]\right)\\
&=\sum_{\mathcal{F} \subseteq E^+(R)\cap E(R') \atop \vbl{\mathcal{F}}=R'}\left(\prod_{e \in \mathcal{F}}(1-\kappa_e\cdot\phi_e(\tau_e))\prod_{e \in E^+(R)\cap E(R') \atop e\not\in \mathcal{F}}\kappa_e\cdot\phi_e(\tau_e)\right),
\end{align*}
where the second equation is due to the conditional independence between $F_e$'s for all $e\in E^+(R)$ given $Y$, and the third equation is due to that $y_e=\tau_e$ for all $e\in E(R')$.

We then only need to verify that the $\kappa_e$'s in above formula are determined only by $(S, \sigma, T,\tau)$, which is obvious because $x_e$ for $e\in E(R')$ is determined by these, since $x_{R}=x_{V\setminus S}=\sigma$ and $x_{R'\setminus R}=x_{S\setminus T}=y_{S\setminus T}=\tau_{S\setminus T}$.
Thus, we have
\begin{align}
\label{eq-fix-event-transform}
\Pr[\mathcal{A}_2 \mid \mathcal{A}_1]  = C_{2,2},
\end{align}
for some $C_{2,2} = C_{2,2}(S, \sigma, T, \tau)$ depends only on $(S, \sigma, T, \tau)$.

\item
Given $Y$, all $F_e$'s are sampled independently. Thus, we have
\begin{align*}
\Pr[\mathcal{A}_3 \mid \mathcal{A}_1 \wedge \mathcal{A}_2] 
&=\Pr[\mathcal{A}_3 \mid \mathcal{A}_1 ]\\
&= \prod_{e \in E^+(R) \setminus E(R')}\kappa_e\cdot \phi_e(y_e) \\
&= \prod_{e \in \delta(R ) \setminus E(R' )} \kappa_e\cdot \phi_e(y_e) \prod_{e \in E(R ) \setminus E(R')} \kappa_e \cdot \phi_e(y_e)\\
&=  \prod_{e \in \delta(R ) \setminus E(R' )} \frac{ \phi_e(y_e) }{ \phi_e(x_e) }
\left(\min_{z\in[q]^e\atop z_{e\setminus S}=x_{e\setminus S}}\phi_e(z)\right)
\prod_{e \in E(R ) \setminus E(R' )} \phi_e(y_e).
\end{align*}
The second equation is due to that $E^+(R ) = E(R ) \cup \delta(R)$. 
The last equation is due to the fact that $\kappa_e = 1$ for all $e \in E(R)$. 
With assumption~\eqref{eq-situation}, the ratios in the last equation have finite values.

Note that  $R = V \setminus S$ and $R' = V \setminus T$. And also note that for any $e \in \delta(R )$, the value of $\min_{z\in[q]^e: z_{e\setminus S}=x_{e\setminus S}}\phi_e(z)$ is determined by $S$ and $\sigma$ because $x_{e\setminus S}=\sigma_{e\setminus S}$. Thus, we have 
\begin{align}
\label{eq-success-event-transform}
\Pr[\mathcal{A}_3 \mid \mathcal{A}_1 \wedge \mathcal{A}_2] &= 	C_{2,3}\prod_{e \in \delta(V \setminus S) \setminus E(V \setminus T ) }\frac{\phi_e(y_e)}{\phi_e(x_e)}\prod_{e \in E(V \setminus S) \setminus E(V \setminus T ) }\phi_e(y_e)\notag\\
&=C_{2,3}\prod_{e \in \delta(S) \cap E^+(T ) }\frac{\phi_e(y_e)}{\phi_e(x_e)}\prod_{e \in E(V \setminus S) \cap E^+(T ) }\phi_e(y_e)
\end{align}
where $C_{2,3}= \prod_{e \in \delta(V \setminus S) \setminus E(V \setminus T ) } \min_{z\in[q]^e: z_{e\setminus S}=x_{e\setminus S}}\phi_e(z)$ depends only on $S,\sigma$ and $T$. The second equation is due to $\delta(V \setminus S) = \delta(S)$ and $E(V \setminus T)=E \setminus E^+(T) $.

\end{itemize}
Combining Equations~\eqref{eq-variable-transform},\eqref{eq-fix-event-transform} and \eqref{eq-success-event-transform}, we have
\begin{align*}
P((x, S), (y, T)) = C_2\prod_{v \in T \setminus S}	\phi_v(y_v)\prod_{e \in \delta(S) \cap E^+(T ) }\frac{\phi_e(y_e)}{\phi_e(x_e)}\prod_{e \in E(V \setminus S) \cap E^+(T) }\phi_e(y_e),
\end{align*}
where $C_2=C_{2,1}C_{2,2}C_{2,3}$  depends only on $(S,\sigma,T,\tau)$. This proves the claim.
\end{proof}

\subsection{A meta algorithm for resampling}\label{sec:meta-algorithm}
Our dynamic sampler can be generalized to a general algorithmic framework for dynamic sampling.
In the following we describe a meta-algorithm for resampling called \GenResample. 
It takes as input a pair $(\boldsymbol{X},\mathcal{R})$ of configuration $\boldsymbol{X}\in [q]^V$ and a resample set $\mathcal{R}$, and returns a new random pair $(\boldsymbol{X}',\mathcal{R}')$. The pseudocode for this meta-algorithm is given in Algorithm~\ref{CodeGeneralResample}.
\begin{algorithm}[ht]
\SetKwInOut{Input}{Input}
\SetKwInOut{Output}{Output}
\Input{a graphical model $\mathcal{I}=(V, E, [q], \Phi)$,  a configuration $\boldsymbol{X}\in[q]^V$ and a $\mathcal{R}\subseteq V$;}
\Output{a new pair $(\boldsymbol{X}',\mathcal{R}')$ of configuration $\boldsymbol{X}'\in[q]^V$ and subset $\mathcal{R}'\subseteq V$;}
		   $\mathcal{R}'' \gets \Expand(\boldsymbol{X}, \mathcal{R})$\;
		$(\boldsymbol{X}',\mathcal{R}') \gets \Resample (\mathcal{I}, \boldsymbol{X}, \mathcal{R}'')$\; 
\Return{$(\boldsymbol{X}',\mathcal{R}')$.}
\caption{\GenResample($\mathcal{I}$, $\boldsymbol{X}$, $\mathcal{R}$)}\label{CodeGeneralResample}
\end{algorithm}

This general resampling procedure consists of two steps: it first expands the current resample set $\mathcal{R}$ to a superset $\mathcal{R}''$ by calling a subroutine $\Expand(\boldsymbol{X}, \mathcal{R})$, which is abstract and may be realized differently  in specific implementations; and then the new pair $(\boldsymbol{X}',\mathcal{R}')$ is constructed by calling \Resample{} (Algorithm~\ref{CodeResample}) on the current sample $\boldsymbol{X}$ and the expanded resample set $\mathcal{R}''$.

The {\em general dynamic sampling algorithm} is obtained by replacing Line~\ref{Incr-while-loop-2} of dynamic sampler (Algorithm~\ref{CodeIncrLRS})  with 
$(\boldsymbol{X},\mathcal{R})\gets\GenResample(\mathcal{I}', \boldsymbol{X}, \mathcal{R})$.

Similar to the definition of $\MC{M}_{\LRes}$ corresponding to the \Resample{} algorithm (Definition~\ref{def:MRes}), a Markov chain $\MC{M}_{\GRes}$ on space $[q]^V\times 2^V$ over states $(X, \mathcal{S})$, where $\mathcal{S}\triangleq V\setminus \mathcal{R}$ stands for the ``sanity set'', can be defined similarly as $\MC{M}_{\LRes}$, by replacing \Resample{} with \GenResample{} in Definition~\ref{def:MRes}.

We also define another chain $\MExp$ for the $\Expand$ subroutine as follows.
Each transition $(X,\mathcal{S})\to (X',\mathcal{S}')$ is given by:
\begin{align*}
\mathcal{R}' &\gets\Expand(X, V\setminus \mathcal{S});\\
\mathcal{S}' &\gets V\setminus  \mathcal{R}'; \\
X' &\gets X.
\end{align*}

The $\MC{M}_{\GRes}$ chain is a composite of the two chains $\MExp$ and $\MLRes$. Therefore, it is obvious that if $\MLRes$ and $\MExp$ both satisfy the equilibrium stated in Condition~\ref{ConMCLV-stationary}, then the composite chain $\MC{M}_{\GRes}$ also satisfies Condition~\ref{ConMCLV-stationary}, which as discussed in Section~\ref{EquilibriumConditions}, would imply the correctness of the general dynamic sampling algorithm with an abstract $\Expand$ subroutine. Meanwhile, by Lemma~\ref{LemMRDetailedBalance} and Lemma~\ref{LemRefineToMacro}, we already has Condition~\ref{ConMCLV-stationary} satisfied by the $\MLRes$ chain. Therefore, the following theorem is true.

\begin{theorem}[\bf Correctness of the general dynamic sampling algorithm]\label{ThmGeneralLRSCorrect} 
Assuming that $\MExp$ satisfies Condition~\ref{ConMCLV-stationary} and the input sample $\boldsymbol{X}\sim\mu_{\mathcal{I}}$, upon termination, the general dynamic sampling algorithm returns a perfect sample $\boldsymbol{X}'\sim\mu_{\mathcal{I}'}$.
\end{theorem}

\noindent
Our dynamic sampler (Algorithm~\ref{CodeIncrLRS}) is a special case of the general dynamic sampling algorithm, where the subroutine $\Expand$ is the trivial one: $\Expand(\boldsymbol{X}, \mathcal{R}) = \mathcal{R}$. It is easy to verify that the conditional Gibbs property is invariant under the chain $\MExp$ induced by this trivial $\Expand$ subroutine, thus this $\MExp$ satisfies Condition~\ref{ConMCLV-stationary}.
The correctness of Algorithm~\ref{CodeIncrLRS} stated in Thorem~\ref{ThmLRSCorrect} follows as a corollary of Theorem~\ref{ThmGeneralLRSCorrect}.

\section{Proof of Fast Convergence}
\label{sec:convergence}
We then analyze the convergence rate of our dynamic sampler on general graphical models.

We first show a contraction behavior of the \Resample{} (Algorithm~\ref{CodeResample}) under condition~\eqref{eq:converge-cond}. Fixed a graphical model $\mathcal{I}$, the resampling procedure \Resample{} takes a pair $(\boldsymbol{X}, \mathcal{R}) \in [q]^V \times 2^V $ as input and returns a pair $(\boldsymbol{X}', \mathcal{R}')  \in [q]^V \times 2^V $ as
\begin{align*}
(\boldsymbol{X}', \mathcal{R}') \gets \Resample(\boldsymbol{X}, \mathcal{R}).	
\end{align*}

We construct an integral-valued {\em potential function}  $\HLRes: 2^V \rightarrow \mathbb{Z}_{\geq 0}$ and show that there is a decay on the value of $\HLRes(\mathcal{R})$.
Let $\mathcal{R}\subseteq V$ be a set of variables. The potential function $\HLRes(\mathcal{R})$ is constructed as the minimum number of constraints that can cover all variables in $\mathcal{R}$. Formally:

\begin{definition}	
\label{DefPotentialFun}
Let $\mathcal{I} = (V, E, [q], \Phi)$ be a graphical model. The potential function $\HLRes: 2^V \rightarrow \mathbb{Z}_{\geq 0}$ is defined as
\begin{align*}
\forall \mathcal{R} \subseteq V: \qquad \HLRes(\mathcal{R}) \triangleq \min \left\{ |\mathcal{F}|:  \mathcal{F} \subseteq E \text{ and } \mathcal{R} \subseteq \bigcup_{e \in \mathcal{F}}e \right\}.
\end{align*}
\end{definition}

\noindent
Without loss of generality, we assume that each variable is incident to at least one constraint so that the above potential function is well-defined. If otherwise, the variable incident to no constraint is independent with all other variables, which can be easily handled separately by the sampler. 

In this definition of potential function, we require that $\mathcal{R} \subseteq \bigcup_{e \in \mathcal{F}}e$ rather than $\mathcal{R} = \bigcup_{e \in \mathcal{F}}e$. This is because in Line~\ref{Incr-initial-R} of Algorithm~\ref{CodeIncrLRS}, the resampling set is constructed as $\mathcal{R} = \vbl{D}$, where $D \subseteq V \cup E$. Since $D$ is allowed to contain some variables in $V$ (for example, $D$ may only contain a single variable), then there may not exist a subset $\mathcal{F} \subseteq E$ such that $\vbl{D} = \bigcup_{e \in \mathcal{F}}e$.

\begin{lemma}
\label{LemStepDecay}
Assume that condition~\eqref{eq:converge-cond} holds for the input graphical model $\mathcal{I}$. Given any $(\boldsymbol{X}, \mathcal{R}) \in [q]^V \times 2^V $, the following holds for the $\mathcal{R}'\subseteq V$ returned by Algorithm~\ref{CodeResample} on the input $(\boldsymbol{X}, \mathcal{R})$:
\begin{align*}
\E{\HLRes(\mathcal{R}')} \leq (1-\delta)\HLRes(\mathcal{R}).
\end{align*}
\end{lemma}
\begin{proof}
In Algorithm~\ref{CodeResample}, the random configuration $X'$ is obtained by sampling $X'_v$ independently for all $v \in \mathcal{R}$ and setting $X'_v = X_v$ for all $v \not\in \mathcal{R}$. 
Then each constraint $e \in E^+(\mathcal{R})$ resamples $F_e \in \{0,1\}$ with $\Pr[F_e = 0] = \kappa_e \cdot \phi_e(X'_e)$. Since $\phi_e: [q]^e \rightarrow [B_e, 1]$, then we have
\begin{align}
\label{eq:Pr_Fe}	
\Pr[F_e = 1] \leq 1 - B_e^2.
\end{align}
The set $\mathcal{R}'$ returned by Algorithm~\ref{CodeResample} is constructed by
$
\mathcal{R}' = \bigcup_{e \in E: F_e = 1}e = \bigcup_{e \in E^+(\mathcal{R}): F_e = 1}e.	
$
By the definition of the potential function, $\HLRes(\mathcal{R}')$ is at most the number constraint $e \in E^+(\mathcal{R})$ such that $F_e= 1$.
 By the linearity of expectation and inequality~\eqref{eq:Pr_Fe}, we have
\begin{align}
\label{eq-bound-H-R'}
\E{\HLRes(\mathcal{R}')} &\leq \sum_{e \in E^+(\mathcal{R})}\Pr[F_e = 1] \leq   \sum_{e \in E^+(\mathcal{R})}(1 - B_e^2).
\end{align}

Let $\mathcal{F} \subseteq E$ be the subset of constraints such that $\mathcal{R} \subseteq \bigcup_{e \in \mathcal{F}}e$ and $\HLRes(\mathcal{R}) = |\mathcal{F}|$. If the choice of $\mathcal{F}$ is not unique, we pick an arbitrary one. 
Since $\mathcal{R} \subseteq \bigcup_{e \in \mathcal{F}}e$, then we have
$
E^+(\mathcal{R}) \subseteq \left(\bigcup_{e \in \mathcal{F}}\Gamma(e)\right) \cup \mathcal{F}.
$
Combining with inequality~\eqref{eq-bound-H-R'}, it holds that
\begin{align}
\label{eq-H-upper-bound}
\E{\HLRes(\mathcal{R}')} &\leq \sum_{e \in \mathcal{F}}\left(1 - B_e^2 + \sum_{f \in \Gamma(e)}(1 - B_f^2)\right).
\end{align}
By condition~\eqref{eq:converge-cond}, it holds that
\begin{align*}
\forall e \in E: \qquad B_e^2 \geq 1 - \frac{1 -\delta}{1+d},
\end{align*}
where $d = \max_{e \in E}|\Gamma(e)|$ is the maximum degree of dependency graph. Thus
\begin{align*}
\E{\HLRes(\mathcal{R}')} \leq \sum_{e \in \mathcal{F}}\left(\frac{1-\delta}{1 +d} + \sum_{f \in \Gamma(e)}\frac{1-\delta}{1+d} \right) \leq 
 (1 - \delta)|\mathcal{F}| = (1 - \delta)\HLRes(\mathcal{R}).
\end{align*}
\end{proof}

\newcommand{\LemmaFirstStepDecay}{\color{blue}
\begin{lemma}
\label{LemFirstStepDecay}
If the input graphical model $\mathcal{I}$ satisfies condition~\eqref{eq:converge-cond}, and assuming $\mathcal{F}\subseteq V \cup E$, then the following holds for the $\mathcal{F}'\subseteq E$ returned by Algorithm~\ref{CodeResample}: 
\begin{align*}
\E{|\mathcal{F}'|} \leq 2|\mathcal{F}|.	
\end{align*}
\end{lemma}
\begin{proof}
Recall that if the input graphical model $\mathcal{I}$ satisfies condition~\eqref{eq:converge-cond}, then it must hold that
\begin{align*}
\forall e\in E,\quad \sum_{f\in\Gamma(e)}(1-B_{f}^2)\le B_e-\delta.	
\end{align*}
For all vertices $v \in V$, we prove following inequality
\begin{align}
\label{eq-condition-vertex}
\forall v\in V:\quad \sum_{f \in \Gamma(v)}(1 - B_f^2) \leq 2.	
\end{align}
If $\Gamma(v) = \emptyset$, then inequality~\eqref{eq-condition-vertex} holds trivially. If otherwise, we pick an arbitrary edge $e \in \Gamma(v)$. Since $v \in e$, then it must hold that $\Gamma(v) \subseteq \Gamma(e) \cup \{e\}$. By condition~\eqref{eq:converge-cond}, we have 
\begin{align*}
\sum_{f \in \Gamma(v)}(1 - B_f^2) \leq \sum_{f \in \Gamma(e) \cup \{e\}}(1-B_f^2) \leq B_e-\delta + 1 - B_e^2 \leq 2.	
\end{align*}
This proves inequality~\eqref{eq-condition-vertex}.

In Algorithm~\ref{CodeResample}, after resampling of all variables in $\vbl{\mathcal{F}}$, each $e\in\mathcal{F} \cap E$ joins $\mathcal{F}'$ with probability $1-\phi_e(X_e)\le 1-B_e$, and each $f\in\Gamma(\mathcal{F})$ joins $\mathcal{F}'$ with probability $1-B_f\frac{\phi_f(X_f')}{\phi_f(X_f)}\le 1-B_f^2$. By linearity of expectation, we have
\begin{align*}
\E{|\mathcal{F}'|} &= \sum_{e \in \mathcal{F} \cap E}\Pr[e \in \mathcal{F}'] + \sum_{f \in \Gamma(\mathcal{F})}\Pr[f \in \mathcal{F}']\\
&\leq  \sum_{e \in \mathcal{F} \cap E}(1-B_e) + \sum_{f \in \Gamma(\mathcal{F})}(1-B_f^2)\\
&\leq \sum_{e \in \mathcal{F} \cap E}\left(1 -B_e + \sum_{f \in \Gamma(e)}(1 - B_f^2)\right) + \sum_{v \in \mathcal{F} \cap V}\sum_{f \in \Gamma(v)}(1 - B^2_f),
\end{align*}
where the last inequality is due to the fact that $\Gamma(\mathcal{F}) \subseteq \bigcup_{z \in \mathcal{F}}\Gamma(z)$. By condition~\eqref{eq:converge-cond} and inequality~\eqref{eq-condition-vertex}, this quantity is bounded by $\sum_{e \in \mathcal{F}\cap E}(1 - \delta) + 2|\mathcal{F} \cap V| \leq 2 |\mathcal{F}|$.

\end{proof}
}

With this stepwise decay on the potentials, we can now prove the main theorem on fast convergence of the dynamic sampler.
\begin{proof}[Proof of Theorem~\ref{ThmLRSConv}]

Let $\mathcal{R}_0 = \vbl{D}$ and $\boldsymbol{X}_0=\boldsymbol{X}\sim\mu_{\mathcal{I}}$ be the initial sample. For $t\ge 1$, let 
\[
(\boldsymbol{X}_{t},\mathcal{R}_t)=\text{Resample}(\mathcal{I}', \boldsymbol{X}_{t-1}, \mathcal{R}_{t-1}).
\]
Let $T$ be the smallest integer such that $\mathcal{R}_T = \emptyset$. The Algorithm~\ref{CodeIncrLRS} terminates after $T$ iterations.
By Lemma~\ref{LemStepDecay}, for any $t\ge 1$, we have 
\begin{align*}
\E{\HLRes(\mathcal{R}_t)\mid \mathcal{R}_{t-1}}\leq (1-\delta)\HLRes(\mathcal{R}_{t-1}).
\end{align*}
Taking expectation over $\mathcal{R}_{t-1}$ on both sides gives the recurrence:
\begin{align*}
\forall t \geq 1: \qquad\E{\HLRes(\mathcal{R}_t)}\leq (1-\delta)\E{\HLRes(\mathcal{R}_{t-1})}.	
\end{align*}
For the base case, by the definition of potential function, it holds that $\HLRes(\mathcal{R}_0) \leq |D|$.
With the above recurrence, this implies that for all $t\ge 0$:
\begin{align}
\label{eq-IncrBoundFt}
\E{\HLRes(\mathcal{R}_t)} \leq |{D}|(1-\delta)^{t} \leq |{D}|\mathrm{e}^{-t\delta}.	
\end{align}
Let $\ell =\frac{1}{\delta}\ln |D|+1$. Then $|D|\mathrm{e}^{-\ell\delta} < 1$. By Definition~\ref{DefPotentialFun}, it holds that $\HLRes(\mathcal{R}) = 0$ if and only if $\mathcal{R} = \emptyset$. We then bound the expected number of iterations as:
\begin{align*}
\E{T} &= \sum_{t \geq 1} \Pr[T \geq t]\\
&=\sum_{t \geq 1} \Pr[\HLRes(\mathcal{R}_t) \geq 1]\\
&\leq \ell + \sum_{t \geq \ell}|D|\mathrm{e}^{-t\delta} \qquad (\text{Markov inequality})\\
&\leq \ell + \frac{1}{1-\mathrm{e}^{-\delta}}\\
\qquad&=O\left(\log |{D}|\right).
\end{align*}
Hence, the algorithm terminates within $O(\log |D|)$ iterations in expectation.

Recall that $d = \max_{e \in E}|\Gamma(e)|$ is the maximum degree of the dependency graph. Let $k=\max_{e\in E}|e|$ denote the maximum edge size. Then the cost of the $t$-th iteration in Algorithm~\ref{CodeIncrLRS} is at most $O(kd\HLRes(\mathcal{R}_{t-1}))$. This is because the variable subset $\mathcal{R}_{t-1}$ is at most incident to $O(d \HLRes(\mathcal{R}_{t-1}))$ constraints.
By inequality~\eqref{eq-IncrBoundFt}, we have
\begin{align*}
\sum_{t \geq 0}\E{\HLRes(\mathcal{R}_t)} \leq  \sum_{t \geq 0}|D|\mathrm{e}^{-t\delta} = O(|D|).
\end{align*}
Therefore the expected total cost is bounded by $O(kd|D|)$, which is $O(|D|)$ when  the maximum degree $d$ and the maximum constraint size $k$ are both constants.
\end{proof}

\section{Applications to Spin Systems}
\label{sec:applications}
\subsection{Ising model}

Let $\mathcal{I} = (V,E,\underline{\beta})$ be an Ising model on graph $G = (V, E)$,
where each edge $e\in E$  is associated with an \concept{inverse temperature} $\beta_e\in\mathbb{R}$.
The Gibbs distribution over all configurations $\sigma\in\{-1,+1\}^V$ is defined such that 
\[
\mu(\sigma)\propto \prod_{e=(u,v)\in E}\exp(\beta_e\sigma_u\sigma_v).
\]
The Ising models can be expressed as graphical models.
Our dynamic sampler (Algorithm~\ref{CodeIncrLRS}) instantiated on the Ising models gives us the following dynamic perfect Ising sampler (Algorithm~\ref{CodeIsing}).

\begin{algorithm}[ht]
\SetKwInOut{Input}{Input}
\SetKwInOut{Update}{Update}
\SetKwInOut{Output}{Output}
\Input{an Ising model $\mathcal{I}$ and a random sample $\boldsymbol{X}\sim\mu_{\mathcal{I}}$;}
\Update{an update of  vertices and edges in $D\subseteq V\cup{V\choose 2}$  modifying $\mathcal{I}$ to $\mathcal{I}'=(V,E,\underline{\beta})$;}
\Output{a random sample $\boldsymbol{X}\sim\mu_{\mathcal{I}'}$;}
$\mathcal{R}\gets \vbl{D}$\;
\While{
$\mathcal{R}\neq \emptyset$
}{
		all $e\in E$ start with no failure\;
		each $e=(u,v)\in \delta(\mathcal{R})$ fails independently with probability $1-\mathrm{e}^{- |\beta_e|}\cdot\mathrm{e}^{-\beta_e X_u X_v}$\;\label{Ising-fail-step-1}
		for each $v \in  \mathcal{R}$, resample $X_v\in\{-1,+1\}$ uniformly and independently\;\label{Ising-resample-step}
		each non-failed $e =(u, v) \in E^+(\mathcal{R})$ fails independently with prob.~$1-\mathrm{e}^{- |\beta_e|}\cdot\mathrm{e}^{\beta_e X_u X_v}$\;\label{Ising-fail-step-2}
		$\mathcal{R}\gets\bigcup_{\text{failed }e\in E} e$\;	\label{Ising-construct-R}
}
\Return{$\boldsymbol{X}$}\;
\caption{Dynamic Ising Sampler}\label{CodeIsing}
\end{algorithm}

For simplicity of exposition we consider the Ising model without external field. For the Ising model with local fields, at Line~\ref{Ising-resample-step} the random variables are sampled proportional to the local fields, and the algorithm still outputs from the correct Gibbs distribution.

We then prove the fast convergence of this dynamic Ising sampler.
\begin{proof}[Proof of Theorem~\ref{ThmIsing}]
Let $\boldsymbol{X} \in \{-1, +1\}^V$ be a random configuration and $\mathcal{R} \subseteq V$ be a random subset of variables such that $(\boldsymbol{X}, \mathcal{R})$ is conditionally Gibbs with respect to the Ising model $\mathcal{I}$. 
Fix one iteration  of the while loop in Algorithm~\ref{CodeIsing}.
The pair $(\boldsymbol{X}, \mathcal{R})$ is transformed to a random $(\boldsymbol{X}', \mathcal{R}')$ after this iteration. 
Assume that $\exp(-2|\beta_e|) \geq 1 - \frac{1}{\alpha\Delta + 1}$ for all edge $e \in E$, where $\alpha \approx 2.22\ldots$ is the positive root of $ \alpha = \frac{2}{1+ \exp\left(-\frac{1}{\alpha}\right)} + 1$. 
For any fixed $\mathcal{R}$, we show the following holds 
\begin{align*}
\E{\HLRes(\mathcal{R}') \mid \mathcal{R}} \leq \left(1 -\delta(\alpha,\Delta)\right)\HLRes(\mathcal{R}),
\end{align*}
where the potential function $\HLRes(\mathcal{R}) = \min \left\{ |\mathcal{F}|:  \mathcal{F} \subseteq E \text{ and } \mathcal{R} \subseteq \bigcup_{e \in \mathcal{F}}e \right\}$ is as constructed in Definition~\ref{DefPotentialFun} and $\delta(\alpha,\Delta) \triangleq \frac{1}{\alpha \left(1+\exp\left(-\frac{1}{\alpha}\right)\right)(\alpha\Delta + 1)}$.
Sine the above inequality holds for any fixed $\mathcal{R}$, then we have
\begin{align*}
\E{\HLRes(\mathcal{R}')} \leq \left(1 -\delta(\alpha,\Delta)\right)\E{\HLRes(\mathcal{R})}.
\end{align*}
Note that $\delta(\alpha, \Delta) = \Omega(1)$ since $\Delta = O(1)$. The value of the potential function decays with a constant factor in expectation.
Therefore,
Theorem~\ref{ThmIsing} can be proved by going through the same proof as Theorem~\ref{ThmLRSCorrect}.

The set $\mathcal{R}'$ is constructed by taking the union of all $e \in E^+(\mathcal{R})$ that fail in this iteration.
Note that the potential $\HLRes(\mathcal{R}')$ is at most the number of failed edges $e \in E^+(\mathcal{R})$.
 By linearity of expectation,
 \begin{align}
\label{eq-Ising-H-R'}
\E{\HLRes(\mathcal{R}') \mid \mathcal{R}} &\leq \sum_{e \in E^+(\mathcal{R})}\Pr[\, e\text{ fails}\,]
\end{align}
We then bound the probability that an edge $e \in E^+(\mathcal{R})$ fails.
Let $\beta^*$ be the maximum value of $|\beta_e|$ for all $e \in E$, which is defined as
\begin{align*}
\beta^* \triangleq \max_{e \in E}|\beta_e|	.
\end{align*}
First, consider the internal edges $e = (u, v) \in E(\mathcal{R})$.  
Note that for such internal edges, $X'_u,X'_v\in\{-1,+1\}$ are sampled uniformly at Line~\ref{Ising-resample-step} and $e$ only fails at Line~\ref{Ising-fail-step-2} with probability $1-\exp(\beta_eX'_uX'_v - |\beta_e|)$. Therefore,
\begin{align}
\label{eq-Pr-ER}
\forall e \in E(\mathcal{R}): \qquad \Pr[\, e\text{ fails} \,] = \frac{1 -\exp(-2|\beta_e|)}{2} \leq \frac{1 -\exp(-2\beta^*)}{2}.
\end{align}
Consider the boundary edges $e = (u, v) \in \delta(\mathcal{R})$ where $u \in \mathcal{R}$ and $v \not\in \mathcal{R}$. To bound the probability that $e$ fails, we give a lower bound of  $\Pr[X_v = c]$ for any $c \in \{-1,1\}$. Let $N(v) \triangleq \{u \in V \mid (u, v) \in E\}$ be the set of neighbors of vertex $v$ in graph $G$. By the chain rule, we have
\begin{align*}
\forall c \in \{-1, 1\}:	\qquad \Pr[\,X_v = c\,] &= \sum_{\sigma \in [q]^{N(v)}}\Pr[X_{N(v)} = \sigma] \Pr[X_v = c \mid X_{N(v)} =\sigma].	
\end{align*}
Note that the pair $(\boldsymbol{X}, V\setminus\mathcal{R})$ is conditionally Gibbs with respect to $\mathcal{I}$.  Then it holds that
\begin{align*}
\forall \sigma \in [q]^{N(v)}:	\qquad \Pr[X_v = c \mid X_{N(v)} = \sigma]
&=\frac{\prod_{u \in N(v) } \exp(\beta_{(u,v)}\sigma_u c)}{\prod_{u \in N(v) } \exp(\beta_{(u,v)}\sigma_u c) + \prod_{u \in N(v) } \exp(-\beta_{(u,v)}\sigma_u c)} \\
&\geq \frac{\prod_{u \in N(v)}\exp(-|\beta_{(u,v)}|)}{\prod_{u \in N(v)}\exp(-|\beta_{(u,v)}|) + \prod_{u \in N(v)}\exp(|\beta_{(u,v)}|)}\\
&\geq \frac{\exp(-\beta^*\Delta)}{\exp(-\beta^*\Delta) + \exp(\beta^*\Delta)}.
\end{align*}
Hence, for any edge $e = (u, v) \in \delta(\mathcal{R})$ where $u \in \mathcal{R}$ and $v \not\in \mathcal{R}$, we obtain the following bound
\begin{align}
\label{eq-lowerbound-X_v-c}
\forall c \in \{-1, 1\}:	\qquad \Pr[\,X_v = c\,] &\geq 	\frac{\exp(-\beta^*\Delta)}{\exp(-\beta^*\Delta) + \exp(\beta^*\Delta)} = \frac{\exp(-2\beta^*\Delta)}{1+\exp(-2\beta^*\Delta)}.
\end{align}
Note that $X'_v = X_v$ because $X_v$ is not resampled and  the edge $e= (u, v) \in \delta(\mathcal{R})$ fails if it fails either at Line~\ref{Ising-fail-step-1} or at Line~\ref{Ising-fail-step-2}. 
Given $\boldsymbol{X}$ and $\boldsymbol{X}'$ the probability that $e$ fails can be expressed as
\begin{align*}
\forall e \in \delta(\mathcal{R}): \qquad	\Pr[\, e\text{ fails}\mid \boldsymbol{X}, \boldsymbol{X}']  &= 1 -\exp\left(\beta_e X'_u X'_v - \beta_eX_uX_v- 2|\beta_e|\right)\\
&= 1 - \exp\left(\beta_e X_v \left(X'_u - X_u\right)\right) \cdot \exp(-2|\beta_e|).
\end{align*}
Since $X'_u \in \{-1, +1\}$ is sampled uniformly and independently at Line~\ref{Ising-resample-step}, we have $\Pr[X'_u = X_u] = \Pr[X'_u \neq X_u] = \frac{1}{2}$. Combining with inequality~\eqref{eq-lowerbound-X_v-c}, the probability that $e$ fails can be bounded as
\begin{align}
\label{eq-Pr-deltaR}
\forall e \in \delta(\mathcal{R}): \qquad	\Pr[\, e\text{ fails}\,] &\leq 1 - \frac{1}{2}\exp(-2|\beta_e|) 
- \frac{1}{2} \left(\frac{\exp(-2\beta^*\Delta)}{1+\exp(-2\beta^*\Delta)} + \frac{ \exp(-4|\beta_e|)}{1+\exp(-2\beta^*\Delta)}\right)\notag\\
&=\frac{1 - \exp(-2|\beta_e|)}{2}+ \frac{1 - \exp(-4|\beta_e|)}{2(1 + \exp(-2\beta^*\Delta))}\notag\\
&\leq \frac{1 - \exp(-2\beta^*)}{2} + \frac{1 - \exp(-4\beta^*)}{2(1 + \exp(-2\beta^*\Delta))}.
\end{align}
Denote $B \triangleq \exp(-2\beta^*)$. Combining inequalities~\eqref{eq-Pr-ER} and~\eqref{eq-Pr-deltaR}, we have
\begin{align}
\label{eq-Pr-E+R}
\forall e \in E^+(\mathcal{R}): \qquad 	\Pr[\, e\text{ fails}\,] \leq \frac{1 - B}{2} + \frac{1 - B^2}{2(1 + B^{\Delta})} = \frac{(1-B)}{2}\left(1 + \frac{1+B}{1+B^\Delta}\right).
\end{align}
Let $\mathcal{F} \subseteq E$ be a subset of edegs such that $\mathcal{R} \subseteq \bigcup_{e \in \mathcal{F}}e$ and $\HLRes(\mathcal{R}) = |\mathcal{F}|$. If such $\mathcal{F}$ is non-unique, we pick an arbitrary one. 
Since $\mathcal{R} \subseteq \bigcup_{e \in \mathcal{F}}e$, then we have $E^+(\mathcal{R}) \subseteq \left(\bigcup_{e \in \mathcal{F}}\Gamma(e)\right) \cup \mathcal{F}$, where $\Gamma(e)=\{f\in E\mid f\neq e\wedge f\cap e\neq\emptyset\}$ denotes the neighborhood of $e$ in the dependency graph. Hence, the size of the set $E^+(\mathcal{R})$ is at most $(2\Delta -1)|\mathcal{F}|$, because $\max_{e \in E}|\Gamma(e)| = 2(\Delta - 1)$.
Combining this fact with~\eqref{eq-Ising-H-R'} and~\eqref{eq-Pr-E+R}, we have
\begin{align*}
\E{\HLRes(\mathcal{R}')\mid \mathcal{R}} &\leq \sum_{e \in E^+(\mathcal{R})}\Pr[\,e\text{ fails}\,] \\
&\leq \frac{(1-B)}{2}\left(1 + \frac{1+B}{1+B^\Delta}\right)|E^+(\mathcal{R})| \\
&\leq \frac{(2\Delta-1)(1-B)}{2}\left(1 + \frac{1+B}{1+B^\Delta}\right)|\mathcal{F}|.
\end{align*}
By our assumption, $\exp(-2|\beta_e|) \geq 1 - \frac{1}{\alpha\Delta + 1}$ for all edge $e \in E$, where $\alpha \approx 2.22\ldots$ is the positive root of $ \alpha = \frac{2}{1+ \exp\left(-\frac{1}{\alpha}\right)} + 1$, which implies
\begin{align*}
B = \exp(-2\beta^*) = \exp\left(-2\max_{e \in E}|\beta_e| \right) \geq 1 - \frac{1}{\alpha\Delta+1}.
\end{align*}
Note that $|\mathcal{F}| = \HLRes(\mathcal{R})$. We have 
\begin{align*}
\E{\HLRes(\mathcal{R}') \mid \mathcal{R}} &\leq \frac{(2\Delta-1)(1-B)}{2}\left( 1 + \frac{1+B}{1+B^\Delta} \right)\HLRes(\mathcal{R})\\
&\leq \frac{2\Delta-1}{2(\alpha\Delta+1)}\left( 1 + \frac{2- \frac{1}{\alpha\Delta+1}}{1+\left( 1 - \frac{1}{\alpha\Delta+1} \right)^\Delta} \right)\HLRes(\mathcal{R})\\
&\leq \frac{1}{\alpha}\left(1 + \frac{2 - \frac{1}{\alpha\Delta + 1}}{1 + \exp\left(-\frac{1}{\alpha}\right)}\right)\HLRes(\mathcal{R}),
\end{align*}
where the last inequality is due to that $\frac{2\Delta-1}{2(\alpha\Delta+1)} \leq \frac{1}{\alpha}$ and $\left( 1 - \frac{1}{\alpha\Delta+1} \right)^\Delta \geq \exp\left(-\frac{1}{\alpha}\right)$, which can be verified to further equal to
\begin{align*}
&\left(\frac{1}{\alpha}\left( 1 + \frac{2}{1+ \exp\left(-\frac{1}{\alpha}\right)}  \right) - \frac{1}{\alpha \left(1+\exp\left(-\frac{1}{\alpha}\right)\right)(\alpha\Delta + 1)}\right)\HLRes(\mathcal{R})\\
= &\left(1 -\frac{1}{\alpha \left(1+\exp\left(-\frac{1}{\alpha}\right)\right)(\alpha\Delta + 1)}\right)\HLRes(\mathcal{R}),
\end{align*}
due to the definition of $\alpha$.

In conclusion, we prove the desired decay on the potential:
\begin{align*}
\E{\HLRes(\mathcal{R}') \mid \mathcal{R}} &\leq \left(1 -\frac{1}{\alpha \left(1+\exp\left(-\frac{1}{\alpha}\right)\right)(\alpha\Delta + 1)}\right)\HLRes(\mathcal{R})= \left(1 -\delta(\alpha,\Delta)\right)\HLRes(\mathcal{R}).
\end{align*}
The rest of the proof can be done by going through the same proof as Theorem~\ref{ThmLRSCorrect}.
\end{proof}

For the Potts model, the Gibbs distribution is defined over all configurations $\sigma\in[q]^V$ such that
\[
\mu(\sigma)\propto\prod_{e=(u,v)\in E}\exp(\beta_e \cdot (2\delta(\sigma_u,\sigma_v)-1)),
\]
where $\delta(\cdot,\cdot)$ is the Kronecker delta.

The dynamic Ising sampler (Algorithm~\ref{CodeIsing}) can be naturally generalized to the dynamic Potts sampler: at Line~\ref{Ising-resample-step}, each $X_v$ is now resampled from $[q]$ uniformly (or proportional to local fields if there are non-zero fields) and independently, and the failure probabilities at Line~\ref{Ising-fail-step-1} and Line~\ref{Ising-fail-step-2} are changed to $\exp(-\beta_e \cdot (2\delta(\sigma_u,\sigma_v)-1)-|\beta_e|)$ and $\exp(\beta_e \cdot (2\delta(\sigma_u,\sigma_v)-1)-|\beta_e|)$ respectively. 
It is easy to verify that this algorithm is precisely the dynamic sampler (Algorithm~\ref{CodeIncrLRS}) instantiated on the Potts model and the same bounds as in Theorem~\ref{ThmIsing} hold for this dynamic Potts sampler.

\subsection{Hardcore model}
\label{Sec:hardcore}
Let $\mathcal{I} = (V, E, \underline{\lambda})$ be a hardcore model on graph $G = (V, E)$, where each vertex $v\in V$ is associated with a fugacity $\lambda_v>0$.
The Gibbs distribution over all configurations $\sigma \in \{0, 1\}^V$ is defined as 
\begin{align*}
\mu(\sigma)\propto \begin{cases}
 \prod_{v \in I(\sigma)}\lambda_v  & \text{if $I(\sigma)$ is an independent set}\\	
 0 &\text{if $I(\sigma)$ is not an independent set},
 \end{cases}
 \end{align*}
where $I(\sigma) \triangleq \{v \in V \mid \sigma_v = 1\}$. 

The hardcore models can be expressed as graphical models. We consider the meta-algorithm for resampling \GenResample{} (Algorithm~\ref{CodeGeneralResample}), with the following $\Expand$ subroutine:
\begin{align}
\label{eq-Expand-hardcore}
\Expand(\boldsymbol{X}, \mathcal{R}) \triangleq \mathcal{R} \cup \{v \in V \setminus \mathcal{R} \mid \exists u \in \mathcal{R} \text{ s.t. } (u, v) \in E \land X_u = 1 \}.
\end{align}
On the hardcore model this is instantiated as the dynamic perfect hardcore sampler (Algorithm~\ref{CodeHardcore}).
\begin{algorithm}[ht]
\SetKwInOut{Input}{Input}
\SetKwInOut{Update}{Update}
\SetKwInOut{Output}{Output}
\Input{a hardcore model $\mathcal{I}$ and a random sample $\boldsymbol{X}\sim\mu_{\mathcal{I}}$;}
\Update{an update of vertices and edges in $D\subseteq V \cup {V\choose 2}$ modifying $\mathcal{I}$ to $\mathcal{I}'=(V,E,\underline{\lambda})$;}
\Output{a random sample $\boldsymbol{X}\sim\mu_{\mathcal{I}'}$;}
$\mathcal{R}\gets \vbl{D}$\;
\While{
$\mathcal{R}\neq \emptyset$
}{for each $v\in \mathcal{R}$ with $X_v=1$, add all neighbors of $v$ in graph $G$ into $\mathcal{R}$\;
		for each $v \in  \mathcal{R}$, resample $X_v\in\{0,1\}$ independently with $\Pr[X_v = 1] = \frac{\lambda_v}{1+\lambda_v}$\;\label{hardcore-resample-step}
		$\mathcal{R}\gets\bigcup_{e =(u, v)\in E \atop X_u = X_v = 1} e$\;	\label{hardcore-construct-R}
}
\Return{$\boldsymbol{X}$}\;
\caption{Dynamic Hardcore Sampler}\label{CodeHardcore}
\end{algorithm}

The subroutine $\Expand(\boldsymbol{X}, \mathcal{R})$ specified in~\eqref{eq-Expand-hardcore} can be implemented as following.  Each vertex $v \in \mathcal{R}$ with $X_v = 1$, in parallel, adds all of its neighbors in graph $G$ into the set $\mathcal{R}$ to obtained the expanded resample set $\mathcal{R}''=\Expand(\boldsymbol{X}, \mathcal{R})$.

 It is easy to verify that the corresponding Markov chain $\MExp$  satisfies the equilibrium condition. The resample subset $\mathcal{R}'' = \Expand(\boldsymbol{X}, \mathcal{R})$ is fully determined by $\mathcal{R}$ and $X_{\mathcal{R}}$. The subset $\mathcal{R}''$ gives no extra information about $X_{V \setminus \mathcal{R}''}$. Note that  $\mathcal{R} \subseteq \mathcal{R}''$, then $(V \setminus \mathcal{R}'') \subseteq (V \setminus \mathcal{R})$. Therefore, if the pair $(\boldsymbol{X}, V \setminus \mathcal{R})$ is conditionally Gibbs, then the pair $(\boldsymbol{X}, V \setminus \mathcal{R}'')$ must be also conditionally Gibbs. Due to Theorem~\ref{ThmGeneralLRSCorrect}, the above algorithm is correct.

In~\cite{guo2016uniform}, an algorithm is given for sampling hardcore model in a static setting, which can be expressed as our resampling meta-algorithm \GenResample{} with 
$\Expand(\boldsymbol{X}, \mathcal{R}) = \vbl{E^+(\mathcal{R})}$.
Their algorithm can be interpreted as a static version of Algorithm~\ref{CodeHardcore}.  Because after the first iteration of Algorithm~\ref{CodeHardcore}, the pair $(\boldsymbol{X}, \mathcal{R})$ must satisfy that $X_v=1$ for all $v \in \mathcal{R}$, which means the $\Expand(\boldsymbol{X}, \mathcal{R})$ in~\eqref{eq-Expand-hardcore} is exactly $\vbl{E^+(\mathcal{R})}$.

\begin{proof}[Proof of Theorem~\ref{ThmHardcore}]
For this algorithm, we define a new potential function $\HHC: 2^V \rightarrow \mathbb{Z}_{\geq 0}$ as 
\begin{align*}
\forall \mathcal{R} \subseteq V: \qquad \HHC(\mathcal{R}) \triangleq |E(\mathcal{R})|. 
\end{align*}
Let $\boldsymbol{X} \in \{0, 1\}^V$ be a configuration of hardcore model and $\mathcal{R} \subseteq V$ be a subset of vertices such that $(\boldsymbol{X}, \mathcal{R})$ is conditionally Gibbs with respect to the hardcore model $\mathcal{I}$. Fix one iteration of the while loop in Algorithm~\ref{CodeHardcore}. The pair $(\boldsymbol{X}, \mathcal{R})$ is transformed to random pair $(\boldsymbol{X}', \mathcal{R}')$ after this iteration.
Assume that $\lambda_v\leq \frac{1}{\sqrt{2}\Delta-1}$ for all $v\in V$. For any fixed $\mathcal{R}$ and any fixed $X_{\mathcal{R}}$, we show the following holds
\begin{align*}
\E{\HHC(\mathcal{R}') \mid \mathcal{R}, X_{\mathcal{R}}} \leq \left(1 - \frac{1}{2\Delta}\right)\HHC(\mathcal{R}).
\end{align*}
Since the above inequality holds for any fixed $\mathcal{R}$ and  $X_{\mathcal{R}}$, then we have
\begin{align*}
\E{\HHC(\mathcal{R}')} \leq \left(1 - \frac{1}{2\Delta}\right)\E{\HHC(\mathcal{R})}.
\end{align*}
Recall $\Delta = O(1)$. The value of the potential function decays with a constant factor in expectation.
Therefore,
Theorem~\ref{ThmHardcore} can be prove by going through the same proof as Theorem~\ref{ThmLRSCorrect}.

According to Algorithm~\ref{CodeHardcore}, any edge $(u,v) \in E$ belongs to $E(\mathcal{R}')$ if and only if the values of $X'_u$ and $X'_v$ are both resampled as 1. Thus, we have 
\begin{align*}
\forall (u,v) \in E: \qquad (u, v) \in E(\mathcal{R}') \quad\Longleftrightarrow\quad X'_u = X'_v = 1.
\end{align*}
Therefore, to bound the expectation of $\HHC(\mathcal{R}') = |E(\mathcal{R}')|$, we can bound the probability of $X'_u = X'_v = 1$ for each edge $(u,v) \in E^+(\mathcal{R}'')$, where $\mathcal{R}'' = \Expand(\boldsymbol{X}, \mathcal{R})$ is the expanded resample set. Recall $\mathcal{R}''$ depends only on $\mathcal{R}$ and $X_{\mathcal{R}}$. Hence, the set $\mathcal{R}''$ is fixed.

Consider the edge $(u, v) \in E(\mathcal{R}'')$. Since $X'_u$ and $X'_v$ are sampled independently, then we have
\begin{align}
\label{eq-ER''-Pr}
\forall (u, v) \in E(\mathcal{R}''): \qquad \Pr[(u,v) \in E(\mathcal{R}')] = \Pr[X'_u = 1 \wedge X'_v= 1 ]=\frac{\lambda_u\lambda_v}{(1 + \lambda_u)(1 + \lambda_v)}.
\end{align}
Consider the edge $(u, v) \in \delta(\mathcal{R}'')$ where $u \in \mathcal{R}''$ and $v \not\in \mathcal{R}''$. Since $X'_u$ is sampled independently and $X'_v = X_v$, then we have
\begin{align}
\label{eq-deltaR''-Pr-1}
\forall (u, v) \in \delta(\mathcal{R}'') \text{ where } u \in \mathcal{R}'' \land v \not\in \mathcal{R}'': \qquad \Pr[(u,v) \in E(\mathcal{R}')] = \frac{\lambda_u}{1 + \lambda_u}\Pr[X_v = 1]. 	
\end{align}
Note that the pair $(\boldsymbol{X}, \mathcal{R}'')$ is conditionally Gibbs with respect to $\mathcal{I}$. By the definition of $\Expand$ in~\eqref{eq-Expand-hardcore}, for any edge $(u, v) \in \delta(\mathcal{R}'')$ where $u \in \mathcal{R}''$ and $v \not\in \mathcal{R}''$, it must hold that $X_u = 0$. Thus $X_{V \setminus \mathcal{R}''}$ is a random configuration sampled from the distribution of the hardcore model on induced subgraph $G[V \setminus \mathcal{R}'']$. We denote such distribution as $\mu_{V \setminus \mathcal{R}''}$. 
For any vertex $v \in V \setminus \mathcal{R}''$, suppose $\sigma \in \{0, 1\}^{V \setminus \mathcal{R}''}$ is an independent set on subgraph $G[V \setminus \mathcal{R}'']$ with $\sigma_v = 1$. Then $\sigma' \in \{0, 1\}^{V \setminus \mathcal{R}''}$ with $\sigma'_u = \sigma_u$ for $u \in V \setminus (\mathcal{R}'' \cup \{v\})$ and $\sigma'_v = 0$ is also an independent set on subgraph $G[V \setminus \mathcal{R}'']$. Note that $\frac{\mu_{V \setminus \mathcal{R}''}(\sigma)}{\mu_{V \setminus \mathcal{R}''}(\sigma')} = \lambda_v$. We have
 $\frac{\Pr[X_v = 1]}{\Pr[X_v = 0]} \leq \frac{\lambda_v}{1}$, which implies
\begin{align}
\label{eq-deltaR''-Pr-2}
\forall v \in V \setminus \mathcal{R}'' \qquad \Pr[X_v = 1] \leq \frac{\lambda_v}{1+\lambda_v}.	
\end{align}

Let $\lambda = \max_{v \in V} \lambda_v$. Combining~\eqref{eq-ER''-Pr},~\eqref{eq-deltaR''-Pr-1} and~\eqref{eq-deltaR''-Pr-2}, we have
\begin{align*}
\E{\HHC(\mathcal{R}')\mid \mathcal{R}, X_{\mathcal{R}}} \leq \sum_{e = (u,v) \in E^+(\mathcal{R}'')}\frac{\lambda_u\lambda_v}{(1 + \lambda_u)(1 + \lambda_v)} \leq \sum_{e = (u,v) \in E^+(\mathcal{R}'')} \left( \frac{\lambda}{1+\lambda} \right)^2.
\end{align*}
By the definition of $\Expand$, it holds that
$
| E(\mathcal{R''}) | \leq (2\Delta - 1)|E(\mathcal{R})|.
$
This is because one edge at most incident to $2\Delta - 2$ edges in graph $G$, where $\Delta$ is the maximum degree of graph $G$. Similarly, the number of edges in $\delta(\mathcal{R}'')$ can be bounded as
$
| \delta(\mathcal{R}'')| \leq (2\Delta - 1)(\Delta - 1) | E(\mathcal{R}) |.
$
Thus
\begin{align*}
\E{\HHC(\mathcal{R}')\mid \mathcal{R}, X_{\mathcal{R}}} \leq \Delta(2\Delta - 1)\left( \frac{\lambda}{1+\lambda} \right)^2 \HHC(\mathcal{R}). 
\end{align*}
Since $\lambda_v \leq \frac{1}{\sqrt{2}\Delta - 1}$ for all $v \in V$, we have $\lambda \leq \frac{1}{\sqrt{2}\Delta - 1}$, which implies
\begin{align*}
\E{\HHC(\mathcal{R}')\mid \mathcal{R}, X_{\mathcal{R}}} \leq \Delta(2\Delta - 1)\left( \frac{1}{\sqrt{2}\Delta} \right)^2 \HHC(\mathcal{R}) \leq \left( 1- \frac{1}{2\Delta} \right)\HHC(\mathcal{R}). 
\end{align*}

The rest can be done by going through the same proof as Theorem~\ref{ThmLRSCorrect}.
\end{proof}

\section{Conclusion and Future Work}\label{sec:conclusion}
We give a dynamic sampling method that allows us to sample perfectly from a broad class of graphical models, while variables and constraints of the graphical model are changing dynamically.
We provide sufficient conditions under which such algorithms run incrementally in time proportional to the size of the update. 
A key to these results is to establish certain equilibrium condition satisfied by local resampling.
On specific graphical models, this equilibrium condition also helps to obtain better convergence of the algorithm.

A major open problem is to give a dynamic sampler for graphical models with ``truly repulsive'' hard constraints, e.g.~uniform proper $q$-coloring.
This requires to overcome certain barrier of the current techniques.

Another direction is on specific graphical models, to improve the regimes for efficient dynamic sampling to the uniqueness regimes.
For example, for the hardcore model, such result would give an efficient algorithm for sampling from the hardcore model in the uniqueness regime, on all graphs including those with unbounded maximum degree, which remains to be open even for static and approximate sampling. So far we only have efficient static and approximate sampling algorithms for graphs with bounded maximum degree~\cite{weitz2006counting} or graphs with large girth and sufficiently large degree~\cite{efthymiou2016convergence}.

Along this direction, a very interesting open problem is to give dynamic samplers from Gibbs distributions with mild decay of correlation. One major open problem is dynamically sampling uniform matchings, which has a decay of correlation with rate $1-O(1/\sqrt{\Delta})$~\cite{bayati2007simple}. A fast dynamic sampler for matchings with, say $O(\Delta^{1.5})$ incremental cost per each update of an edge, even being used as a static and approximate sampler, would improve the $\tilde{O}(n^2m)$ time bound of the Jerrum-Sinclair chain for matchings~\cite{jerrum2003counting}.

More broadly, our techniques should be of independent interest and, in particular, should be useful to extend our results to sampling from joint distributions over continuous distributions and/or with global constraints.

\end{document}